\documentclass[twoside]{aiml14}

\usepackage{aiml14macro}

\usepackage{notation}

\def\lastname{French, Hales and Tay}

\begin{document}

\begin{frontmatter}
  \title{A composable language for action models}
  \author{Tim French}
  \author{James Hales}\footnote{Acknowledges the support of the Prescott Postgraduate Scholarship.}
  \author{Edwin Tay}
  \address{Computer Science and Software Engineering\\The University of Western
Australia\\Perth, Australia}
  
  \begin{abstract}
  Action models are semantic structures similar to Kripke models that represent
  a change in knowledge in an epistemic setting. Whereas the language of action
  model logic~\cite{baltag1998,baltag2005} embeds the semantic structure of an
  action model directly within the language, this paper introduces a language
  that represents action models using syntactic operators inspired by 
  relational actions~\cite{vanditmarsch1999,vanditmarsch2001,vanditmarsch2002}.
  This language admits an intuitive description of the action
  models it represents, and we show in several settings that it is sufficient
  to represent any action model up to a given modal depth and to represent the
  results of action model synthesis~\cite{hales2013},
  and give a sound and complete axiomatisation in some of these settings.
  \end{abstract}

 \end{frontmatter}

  \section{Introduction}

  Dynamic epistemic logic describes the way knowledge can change in multi-agent
  systems subject to informative actions taking place. For example, if Tim were
  to announce ``I like cats'', then everyone in the room would know the
  proposition {\it Tim likes cats} is true, and furthermore, everybody would
  know that this fact is common knowledge among the people in the room. This
  simple informative action is what is referred to as a public announcement
  \cite{plaza1989}, and such actions of these have been extensively studied in
  epistemic logics. More complex actions can include private announcements
  (where some agents are oblivious to the informative action occurring), or a
  group announcement (where members of a group simultaneously make a truthful
  announcement to every other member of the group \cite{agotnes2010}). These
  complex actions may be modelled and reasoned about using action models
  \cite{baltag1998} which are effectively a semantic model of the change caused
  by an informative action. Consequently they are very useful for reasoning
  about the consequences of an informative action, but less well suited to
  reasoning about the action itself.
  
  We present a language for describing epistemic actions syntactically.
  Complex actions may be built as an expression upon simpler primitive actions.
  This approach is a generalisation of the relational actions introduced by van
  Ditmarsch \cite{vanditmarsch2001}. We show in several settings that this
  language is sufficient to represent any informative action represented by an
  action model (up to a given model depth), we present a synthesis result,
  and give a a sound and complete axiomatisation for some of these settings. 
  The synthesis result is an important application of this work: given a
  desired state of knowledge among a group of agents, we are able to compute a
  complex informative action that will achieve that particular knowledge state
  (given it is consistent with the current knowledge of agents). We provided
  these results in a variety of modal logics suited to epistemic reasoning:
  \classK{}, \classKFF{} and \classS{}.
  
  \begin{example}\label{grant-example}
  James, Ed and Tim submit a research grant proposal, and eagerly await the
  outcome. Is there a series of actions that will result in: 

  \begin{enumerate}
    \item Ed knowing the grant application was successful; 
    \item James not knowing whether the grant application was successful, but knowing that
        either Ed or Tim does know;
    \item Tim does not know whether the grant application was successful, but knows that if
        the grant application was unsuccessful, then James knows that it was unsuccessful.
  \end{enumerate}

  Such an epistemic state may be achieved by a series of messages: Ed is sent a
  message congratulating him on a successful application, James is sent a
  message informing him that at least one applicant on each grant has been
  informed of the outcome, and Tim is sent a message informing him that the
  first investigator of all unsuccessful grants has been notified.
  \end{example}
  
  After establishing some technical preliminaries (Section~\ref{technical-preliminaries}) we
  present a syntactic approach for describing informative actions (Sections~\ref{syntax}~and~\ref{semantics}), 
  provide a sound and complete axiomatisation of the language (Section~\ref{axiomatisation}) and
  provide a correspondence result between this language and action models (Section~\ref{correspondence}),
  give a computational method for synthesising actions to achieve an epistemic goal (Section~\ref{synthesis}).

  \section{Technical Preliminaries}\label{technical-preliminaries}

  We recall definitions from modal logic, the action model logic of Baltag,
  Moss and Solecki~\cite{baltag1998,baltag2005} the refinement modal logic of
  van Ditmarsch, French and Pinchinat~\cite{vanditmarsch2009,vanditmarsch2010}
  and the arbitrary action model logic of Hales~\cite{hales2013}.

  Let \atoms{} be a non-empty, countable set of propositional atoms, and let
  \agents{} be a non-empty, finite set of agents.

  \begin{definition}[Kripke model]\label{kripke-model}
  A {\em Kripke model} $\model = \modelTuple$ consists of 
  a {\em domain}  \states{}, which is a non-empty set of states 
  (or possible worlds), 
  an {\em accessibility} function 
  $\accessibility : \agents \to \powerset(\states \times \states)$, 
  which is a function from agents to accessibility relations on \states{}, 
  and a {\em valuation} function $\valuation : \atoms \to \powerset(\states)$,
  which is a function from states to sets of propositional atoms.

  The {\em class of all Kripke models} is called \classK{}. 
  A {\em multi-pointed Kripke model} $\pointedModel{\statesT} = (\model,
  \statesT)$ consists of a Kripke model \model{} along with a designated set of
  states $\statesT \subseteq \states$.
  \end{definition}

  We write $\accessibilityAgent{\agentA}$ to denote $\accessibility(\agentA)$.
  Given two states $\stateS, \stateT \in \states$, 
  we write $\stateS \accessibilityAgent{\agentA} \stateT$ to denote that 
  $(\stateS, \stateT) \in \accessibilityAgent{\agentA}$. 
  We write $\statesT \accessibilityAgent{\agentA}$ to denote the set of states 
  $\{\stateS \in \states \mid \stateT \in \statesT, \stateT \accessibilityAgent{\agentA} \stateS \}$
  and write $\accessibilityAgent{\agentA} \statesT$ to denote the set of states 
  $\{\stateS \in \states \mid \stateT \in \statesT, \stateS \accessibilityAgent{\agentA} \stateT \}$.
  We write $\pointedModel{\stateS}$ as an abbreviation for 
  $\pointedModel{\{\stateS\}}$, and write $\stateT \accessibilityAgent{\agentA}$
  and $\accessibilityAgent{\agentA} \stateT$ as abbreviations for 
  $\{\stateT\} \accessibilityAgent{\agentA}$ and 
  $\accessibilityAgent{\agentA} \{\stateT\}$ respectively.
  As we will often be required to discuss several models at once, we will use
  the convention that 
  $\pointedModel{\statesT} = \pointedModelTuple{\statesT}$,
  $\pointedModel[\prime]{\statesT[\prime]} = \pointedModelTuple[\prime]{\statesT[\prime]}$,
  $\pointedModel[\gamma]{\statesT[\gamma]} = \pointedModelTuple[\gamma]{\statesT[\gamma]}$,
  etc.

  \begin{definition}[Action model]\label{action-model}
  Let \lang{} be a logical language.
  An {\em action model} $\actionModel = \actionModelTuple$ with preconditions defined
  on \lang{} consists of a {\em domain} \actionStates, 
  which is a non-empty, finite set of action points, 
  an {\em accessibility} function $\actionAccessibility : \agents \to \powerset(\actionStates \times \actionStates)$,
  which is a function from agents to accessibility relations on \actionStates, 
  and a {\em precondition} function $\actionPrecondition : \actionStates \to \lang$,
  which is a function from action points to formulae from \lang{}.

  The {\em class of all action models} is called \classAM{}.
  A {\em multi-pointed action model}
  $\pointedActionModel{\actionStatesT} = (\actionModel, \actionStatesT)$
  consists of an action model $\actionModel$ along with a designated set of
  action points $\actionStatesT \subseteq \actionStates$.
  \end{definition}

  We use the same abbreviations and conventions for action models as are used
  for Kripke models. We use the convention of using sans-serif fonts for action
  models, as in \pointedActionModel{\actionStatesT} and italic fonts for Kripke
  models, as in \pointedModel{\statesT}.

  In addition to the class \classK{} of all Kripke models, 
  and the class \classAM{} of all action models 
  we will be referring to several other classes of Kripke models and action models.

  \begin{definition}[Classes of Kripke models and action models]
      The class of all Kripke models / action models with transitive and Euclidean accessibility relations is called \classKFF{} / $\classAM_\classKFF{}$.

      The class of all Kripke models / action models with serial, transitive and Euclidean accessibility relations is called \classKD{} / $\classAM_\classKD{}$.

      The class of all Kripke models / action models with reflexive, transitive and Euclidean accessibility relations is called \classS{} / $\classAM_\classS{}$.
  \end{definition}

  \begin{definition}[Language of arbitrary action model logic]\label{aml-syntax}
  The language \langAaml{} of arbitrary action model logic is inductively defined as:
  \begin{eqnarray*}
      \phi &::=& \atomP \mid 
             \neg \phi \mid
             (\phi \land \phi) \mid
             \necessary[\agentA] \phi \mid
             \allacts{\pointedActionModel{\actionStatesT}} \phi \mid
             \allrefs \phi
  \end{eqnarray*}
  where $\atomP \in \atoms$, $\agentA \in \agents$, and
  $\pointedActionModel{\actionStatesT} \in \classAM$ is a multi-pointed action
  model with preconditions defined on the language \langAaml{}.
  \end{definition}

  We use all of the standard abbreviations for propositional logic, in addition
  to the abbreviations 
  $\possible[\agentA] \phi ::= \neg \necessary[\agentA] \neg \phi$,
  $\someacts{\pointedActionModel{\actionStatesT}} \phi ::= \neg \allacts{\pointedActionModel{\actionStatesT}} \neg \phi$,
  and $\somerefs \phi ::= \neg \allrefs \neg \phi$.

  We also use the cover operator of Janin and Walukiewicz~\cite{janin1995},
  following the definitions given by B{\'i}lkov{\'a}, Palmigiano and
  Venema~\cite{bilkova2008}.  The cover operator, $\covers_\agentA \Gamma$ is
  an abbreviation defined by $\covers_\agentA \Gamma ::= \necessary[\agentA] \bigvee_{\gamma \in \Gamma} \gamma \land \bigwedge_{\gamma \in \Gamma} \possible[\agentA] \gamma$,
  where $\Gamma \subseteq \langAaml$ is a finite set of formulae.
  We note that the modal operators $\necessary[\agentA]$, $\possible[\agentA]$ and $\covers_\agentA$ are interdefineable
  as $\necessary[\agentA] \phi \iff \covers_\agentA \{\phi\} \lor \covers_\agentA \emptyset$ and $\possible[\agentA] \phi \iff \covers_\agentA \{\phi, \top\}$.
  This is the basis for the axiomatisations of refinement modal logic and
  arbitrary action model logic, and plays an important part in our correspondence
  and synthesis results.
  This was previously used as the basis of several axiomatisations
  of refinement modal logics~\cite{vanditmarsch2010,hales2011a,hales2011b,hales2012,bozzelli2012a,hales2013}.

  We refer to the language \langAml{} of action model logic, which is \langAaml{} without the $\allrefs$ operator,
  the language \langRml{} of refinement modal logic, which is \langAaml{} without the $\allacts{\pointedActionModel{\actionStatesT}}$ operator,
  the language \lang{} of modal logic, which is \langAml{} without the $\allacts{\pointedActionModel{\actionStatesT}}$ operator, 
  and the language \langP{} of propositional logic, which is \lang{} without the $\necessary[\agentA]$ operator.

  \begin{definition}[Semantics of modal logic]\label{ml-semantics}
  Let \classC{} be a class of Kripke models and let $\model = \modelTuple \in
  \classC$ be a Kripke model.
  The interpretation of $\phi \in \lang$ in the logic \logicC{} is defined 
  inductively as:
  \begin{eqnarray*}
      \pointedModel{\stateS} \entails \atomP &\text{ iff }& \stateS \subseteq \valuation(\atomP)\\
      \pointedModel{\stateS} \entails \neg \phi &\text{ iff }& \pointedModel{\stateS} \nentails \phi\\
      \pointedModel{\stateS} \entails \phi \land \psi &\text{ iff }& \pointedModel{\stateS} \entails \phi \text{ and } \pointedModel{\stateS} \entails \psi\\
      \pointedModel{\stateS} \entails \necessary[\agentA] \phi &\text{ iff }& \text{for every } \stateT \in \stateS \accessibilityAgent{\agentA} : \pointedModel{\stateT} \entails \phi\\
      \pointedModel{\statesT} \entails \phi &\text{ iff }& \text{for every } \stateT \in \statesT : \pointedModel{\stateT} \entails \phi
  \end{eqnarray*}
  \end{definition}

  \begin{definition}[Bisimilarity of Kripke models]
      Let $\model = \modelTuple \in \classK$ 
      and $\model[\prime] = \modelTuple[\prime] \in \classK$
      be Kripke models. 
      A non-empty relation $\bisimulation \subseteq \states \times \states[\prime]$
      is a {\em bisimulation} if and only if for every $\agentA \in \agents$ 
      and $(\stateS, \stateS[\prime]) \in \bisimulation$ the following conditions hold:

      \paragraph{atoms}
      For every $\atomP \in \atoms$: $\stateS \in \valuation(\atomP)$ if and only if $\stateS[\prime] \in \valuation[\prime](\atomP)$.

      \paragraph{forth-$\agentA$}
      For every $\stateT \in \stateS \accessibilityAgent{\agentA}$ 
      there exists $\stateT[\prime] \in \stateS[\prime] \accessibilityAgent[\prime]{\agentA}$
      such that $(\stateT, \stateT[\prime]) \in \bisimulation$.

      \paragraph{back-$\agentA$}
      For every $\stateT[\prime] \in \stateS[\prime] \accessibilityAgent[\prime]{\agentA}$
      there exists $\stateT \in \stateS \accessibilityAgent{\agentA}$ 
      such that $(\stateT, \stateT[\prime]) \in \bisimulation$.

      If $(\stateS, \stateS[\prime]) \in \bisimulation$ then we call
      $\pointedModel{\stateS}$ and $\pointedModel[\prime]{\stateS[\prime]}$
      {\em bisimilar} and write 
      $\pointedModel{\stateS} \bisimilar \pointedModel[\prime]{\stateS[\prime]}$.
  \end{definition}

  \begin{proposition}
      The relation $\bisimilar$ is an equivalence relation on Kripke models.
  \end{proposition}

  \begin{proposition}
      Let $\pointedModel{\stateS}, \pointedModel[\prime]{\stateS[\prime]} \in \classK$ be Kripke models such that
      $\pointedModel{\stateS} \bisimilar \pointedModel[\prime]{\stateS[\prime]}$. 
      Then for every $\phi \in \lang$:
      $\pointedModel{\stateS} \entails \phi$ if and only if $\pointedModel[\prime]{\stateS[\prime]} \entails \phi$.
  \end{proposition}

  These are well-known results.

  \begin{definition}[$n$-bisimilarity of Kripke models]
      Let $n \in \mathbb{N}$, 
      and let $\pointedModel{\stateS} = \pointedModelTuple{\stateS} \in \classK$ 
      and $\pointedModel[\prime]{\stateS[\prime]} = \pointedModelTuple[\prime]{\stateS[\prime]} \in \classK$ be Kripke models.
      We say that $\pointedModel{\stateS}$ is {\em $n$-bisimilar} 
      to $\pointedModel[\prime]{\stateS[\prime]}$,
      and write $\pointedModel{\stateS} \bisimilar_n \pointedModel[\prime]{\stateS[\prime]}$,
      if and only if for every $\agentA \in \agents$ the following conditions hold:

      \paragraph{atoms}
      For every $\atomP \in \atoms$: $\stateS \in \valuation(\atomP)$ if and only if $\stateS[\prime] \in \valuation[\prime](\atomP)$.

      \paragraph{forth-$n$-$\agentA$}
      If $n > 0$ then
      for every $\stateT \in \stateS \accessibilityAgent{\agentA}$ 
      there exists $\stateT[\prime] \in \stateS[\prime] \accessibilityAgent[\prime]{\agentA}$
      such that $\pointedModel{\stateT} \bisimilar_{(n - 1)} \pointedModel[\prime]{\stateT[\prime]}$ 

      \paragraph{back-$n$-$\agentA$}
      If $n > 0$ then
      for every $\stateT[\prime] \in \stateS[\prime] \accessibilityAgent[\prime]{\agentA}$
      there exists $\stateT \in \stateS \accessibilityAgent{\agentA}$ 
      such that $\pointedModel{\stateT} \bisimilar_{(n - 1)} \pointedModel[\prime]{\stateT[\prime]}$ 
  \end{definition}

  \begin{definition}[Modal depth]
      Let $\phi \in \lang$. The {\em modal depth of $\phi$}, written as $d(\phi)$, is defined recursively as follows:
      \begin{eqnarray*}
          d(\atomP) &=& 0 \text{ for } \atomP \in \atoms\\
          d(\neg \psi) &=& d(\psi)\\
          d(\psi \land \chi) &=& max(d(\psi), d(\chi))\\
          d(\necessary[\agentA] \psi) &=& 1 + d(\psi)
      \end{eqnarray*}
  \end{definition}

  \begin{proposition}
      The relation $\bisimilar_n$ is an equivalence relation on Kripke models.
  \end{proposition}

  \begin{proposition}
      Let $n \in \mathbb{N}$ and 
      let $\pointedModel{\stateS}, \pointedModel[\prime]{\stateS[\prime]} \in \classK$ be Kripke models 
      such that $\pointedModel{\stateS} \bisimilar_n \pointedModel[\prime]{\stateS[\prime]}$. 
      If $m < n$ then $\pointedModel{\stateS} \bisimilar_m \pointedModel[\prime]{\stateS[\prime]}$. 
  \end{proposition}

  \begin{proposition}
      Let $n \in \mathbb{N}$ and 
      let $\pointedModel{\stateS}, \pointedModel[\prime]{\stateS[\prime]} \in \classK$ be Kripke models such that
      $\pointedModel{\stateS} \bisimilar_n \pointedModel[\prime]{\stateS[\prime]}$. 
      Then for every $\phi \in \lang$ such that $d(\phi) \leq n$:
      $\pointedModel{\stateS} \entails \phi$ if and only if $\pointedModel[\prime]{\stateS[\prime]} \entails \phi$.
  \end{proposition}

  \begin{proposition}
      Let $\pointedModel{\stateS}, \pointedModel[\prime]{\stateS[\prime]} \in \classK$ be Kripke models.
      Then $\pointedModel{\stateS} \bisimilar \pointedModel[\prime]{\stateS[\prime]}$
      if and only if for every $n \in \mathbb{N}$: 
      $\pointedModel{\stateS} \bisimilar_n \pointedModel[\prime]{\stateS[\prime]}$.
  \end{proposition}
  
  These are well-known results.

  \begin{definition}[$\agentsB$-bisimilarity of Kripke models]
      Let $\pointedModel{\stateS} = \pointedModelTuple{\stateS} \in \classK$ 
      and $\pointedModel[\prime]{\stateS[\prime]} = \pointedModelTuple[\prime]{\stateS[\prime]} \in \classK$
      be Kripke models. 
      We say that $\pointedModel{\stateS}$ is {\em $\agentsB$-bisimilar}
      to $\pointedModel[\prime]{\stateS[\prime]}$,
      and write $\pointedModel{\stateS} \bisimilar_\agentsB \pointedModel[\prime]{\stateS[\prime]}$,
      if and only if for every $\agentB \in \agentsB$ the following conditions hold:

      \paragraph{atoms}
      For every $\atomP \in \atoms$: $\stateS \in \valuation(\atomP)$ if and only if $\stateS[\prime] \in \valuation[\prime](\atomP)$.

      \paragraph{forth-$\agentB$}
      For every $\stateT \in \stateS \accessibilityAgent{\agentB}$ 
      there exists $\stateT[\prime] \in \stateS[\prime] \accessibilityAgent[\prime]{\agentB}$
      such that $\pointedModel{\stateT} \bisimilar \pointedModel[\prime]{\stateT[\prime]}$.

      \paragraph{back-$\agentB$}
      For every $\stateT[\prime] \in \stateS[\prime] \accessibilityAgent[\prime]{\agentB}$
      there exists $\stateT \in \stateS \accessibilityAgent{\agentB}$ 
      such that $\pointedModel{\stateT} \bisimilar \pointedModel[\prime]{\stateT[\prime]}$ .
  \end{definition}

  \begin{definition}[$\agentsB$-restricted formulae]\label{b-restricted-formulae}
      Let $\agentsB \subseteq \agents$. A {\em $\agentsB$-restricted formula} is defined by the following abstract syntax:
      $$
      \phi ::= \atomP \mid
               \neg \phi \mid
               (\phi \land \phi) \mid
               \necessary[\agentB] \psi
      $$
      where $\atomP \in \atoms$, $\agentB \in \agentsB$, $\psi \in \lang$.
  \end{definition}

  \begin{proposition}
      Let $\agentsB \subseteq \agents$,
      and $\pointedModel{\stateS}, \pointedModel[\prime]{\stateS[\prime]} \in \classK$ be Kripke models such that
      $\pointedModel{\stateS} \bisimilar_\agentsB \pointedModel[\prime]{\stateS[\prime]}$. 
      Then for every $\phi \in \lang$ such that $\phi$ is is a $\agentsB$-restricted formula:
      $\pointedModel{\stateS} \entails \phi$ if and only if $\pointedModel[\prime]{\stateS[\prime]} \entails \phi$.
  \end{proposition}

  This result is trivial.

  We recall the semantics of action model logic of Baltag, Moss and Solecki~\cite{baltag1998,baltag2005}.

  \begin{definition}[Semantics of action model logic]\label{aml-semantics}
  Let \classC{} be a class of Kripke models, let $\model = \modelTuple \in
  \classC$ be a Kripke model and let $\actionModel \in \classAM$ be an action
  model.

  We first define {\em action model execution}. 
  We denote the result of executing the action model $\actionModel$ 
  on the Kripke model $\model$ as $\model \exec \actionModel$, 
  and we define the result as 
  $\model \exec \actionModel = \model[\prime] = \modelTuple[\prime]$ where:
  \begin{eqnarray*}
      \states[\prime] &=& \{(\stateS, \actionStateS) \mid \stateS \in \states, \actionStateS \in \actionStates, \pointedModel{\stateS} \entails \actionPrecondition(\actionStateS)\}\\
      (\stateS, \actionStateS) \accessibilityAgent[\prime]{\agentA} (\stateT, \actionStateT) &\text{ iff }& \stateS \accessibilityAgent{\agentA} \stateT \text{ and } \actionStateS \actionAccessibilityAgent{\agentA} \actionStateT\\
      (\stateS, \actionStateS) \in \valuation[\prime](\atomP) &\text{ iff }& \stateS \in \valuation(\atomP)
  \end{eqnarray*} 

  We also define {\em multi-pointed action model execution} as 
  $\pointedModel{\statesT} \exec \pointedActionModel{\actionStatesT} = \pointedModel[\prime]{\statesT[\prime]} = \pointedModelTuple[\prime]{\statesT[\prime]} = ((\model \exec \actionModel), (\statesT \times \actionStatesT) \cap \states[\prime])$.

  Then the interpretation of $\phi \in \langAml$ in the logic \logicAmlC{} is
  the same as its interpretation in the modal logic \logicC{} given in
  Definition~\ref{ml-semantics}, with the additional inductive case: 
  \begin{eqnarray*}
      \pointedModel{\stateS} \entails \allacts{\pointedActionModel{\actionStatesT}} \phi &\text{ iff }& \pointedModel{\stateS} \exec \pointedActionModel{\actionStatesT} \in \classC \text{ implies } \pointedModel{\stateS} \exec \pointedActionModel{\actionStatesT} \entails \phi
  \end{eqnarray*}
  \end{definition}

  \begin{definition}[Sequential execution of action models]
      Let $\actionModel, \actionModel[\prime] \in \classAM$.
      We define the {\em sequential execution of $\actionModel$ and $\actionModel[\prime]$} as
      $\actionModel \exec \actionModel[\prime] = \actionModel[\prime\prime] = \actionModelTuple[\prime\prime]$ where:
      \begin{eqnarray*}
          \actionStates[\prime\prime] &=& \actionStates \times \actionStates[\prime]\\
          (\actionStateS, \actionStateS[\prime]) \accessibilityAgent[\prime\prime]{\agentA} (\actionStateT, \actionStateT[\prime]) &\text{ iff }& \actionStateS \actionAccessibilityAgent{\agentA} \actionStateT \text{ and } \actionStateS[\prime] \actionAccessibilityAgent[\prime]{\agentA} \actionStateT[\prime]\\
          \actionPrecondition[\prime\prime]((\actionStateS, \actionStateS[\prime])) &=& \someacts{\pointedActionModel{\actionStateS}} \actionPrecondition[\prime](\actionStateS[\prime])
      \end{eqnarray*}

      We also define {\em sequential action of $\pointedActionModel{\actionStatesT}$ and $\pointedActionModel[\prime]{\actionStatesT[\prime]}$} as
      $\pointedActionModel{\actionStatesT} \exec \pointedActionModel[\prime]{\actionStatesT[\prime]} = \pointedActionModel[\prime\prime]{\actionStatesT[\prime\prime]} = \pointedActionModelTuple[\prime\prime]{\actionStatesT \times \actionStatesT[\prime]}$.
  \end{definition}

  \begin{definition}[Bisimilarity of action models]
      Let $\actionModel = \actionModelTuple \in \classAM$ 
      and $\actionModel[\prime] = \actionModelTuple[\prime] \in \classAM$
      be action models. 
      A non-empty relation $\bisimulation \subseteq \actionStates \times \actionStates[\prime]$
      is a {\em bisimulation} if and only if for every $\agentA \in \agents$ 
      and $(\actionStateS, \actionStateS[\prime]) \in \bisimulation$ the following conditions hold:

      \paragraph{atoms}
      $\proves \actionPrecondition(\actionStateS) \iff \actionPrecondition[\prime](\actionStateS[\prime])$

      \paragraph{forth-$\agentA$}
      For every $\actionStateT \in \actionStateS \actionAccessibilityAgent{\agentA}$ 
      there exists $\actionStateT[\prime] \in \actionStateS[\prime] \actionAccessibilityAgent[\prime]{\agentA}$
      such that $(\actionStateT, \actionStateT[\prime]) \in \bisimulation$.

      \paragraph{back-$\agentA$}
      For every $\actionStateT[\prime] \in \actionStateS[\prime] \actionAccessibilityAgent[\prime]{\agentA}$
      there exists $\actionStateT \in \actionStateS \actionAccessibilityAgent{\agentA}$ 
      such that $(\actionStateT, \actionStateT[\prime]) \in \bisimulation$.

      If $(\actionStateS, \actionStateS[\prime]) \in \bisimulation$ then we call
      $\pointedActionModel{\actionStateS}$ and $\pointedActionModel[\prime]{\actionStateS[\prime]}$
      {\em bisimilar} and write 
      $\pointedActionModel{\actionStateS} \bisimilar \pointedActionModel[\prime]{\actionStateS[\prime]}$.
  \end{definition}

  \begin{proposition}
      The relation $\bisimilar$ is an equivalence relation on action models.
  \end{proposition}

  \begin{proposition}
      Let $\pointedModel{\stateS}, \pointedModel[\prime]{\stateS[\prime]} \in \classK$ be Kripke models such that
      $\pointedModel{\stateS} \bisimilar \pointedModel[\prime]{\stateS[\prime]}$. 
      and let $\pointedActionModel{\actionStateS}, \pointedActionModel[\prime]{\actionStateS[\prime]} \in \classAM$ be action models such that
      $\pointedActionModel{\actionStateS} \bisimilar \pointedActionModel[\prime]{\actionStateS[\prime]}$.
      Then
      $(\pointedModel{\stateS} \exec \pointedActionModel{\actionStateS}) \bisimilar (\pointedModel[\prime]{\stateS[\prime]} \exec \pointedActionModel[\prime]{\actionStateS[\prime]})$.
  \end{proposition}

  \begin{proposition}
      Let $\pointedModel{\stateS} \in \classK$ be a Kripke model
      and let $\pointedActionModel{\actionStateS}, \pointedActionModel[\prime]{\actionStateS[\prime]} \in \classAM$ be action models.
      Then 
      $((\pointedModel{\stateS} \exec \pointedActionModel{\actionStateS}) \exec \pointedActionModel[\prime]{\actionStateS[\prime]}) \bisimilar (\pointedModel{\stateS} \exec (\pointedActionModel{\actionStateS} \exec \pointedActionModel[\prime]{\actionStateS[\prime]}))$.
  \end{proposition}

  These results are shown by Baltag, Moss and Solecki~\cite{baltag1998,baltag2005}.

  \begin{definition}[$n$-bisimilarity of action models]
      Let $n \in \mathbb{N}$, 
      and let $\pointedActionModel{\actionStateS} = \pointedActionModelTuple{\actionStateS} \in \classAM$ 
      and $\pointedActionModel[\prime]{\actionStateS[\prime]} = \pointedActionModelTuple[\prime]{\actionStateS[\prime]} \in \classAM$
      be action models.
      We say that $\pointedActionModel{\actionStateS}$ is {\em $n$-bisimilar} 
      to $\pointedActionModel[\prime]{\actionStateS[\prime]}$,
      and write $\pointedActionModel{\actionStateS} \bisimilar_n \pointedActionModel[\prime]{\actionStateS[\prime]}$, 
      if and only if for every $\agentA \in \agents$ the following conditions hold:

      \paragraph{atoms}
      $\proves \actionPrecondition(\actionStateS) \iff \actionPrecondition[\prime](\actionStateS[\prime])$

      \paragraph{forth-$n$-$\agentA$}
      If $n > 0$ then
      for every $\actionStateT \in \actionStateS \actionAccessibilityAgent{\agentA}$ 
      there exists $\actionStateT[\prime] \in \actionStateS[\prime] \actionAccessibilityAgent[\prime]{\agentA}$
      such that $\pointedActionModel{\actionStateT} \bisimilar_{(n - 1)} \pointedActionModel[\prime]{\actionStateT[\prime]}$ 

      \paragraph{back-$n$-$\agentA$}
      If $n > 0$ then
      for every $\actionStateT[\prime] \in \actionStateS[\prime] \actionAccessibilityAgent[\prime]{\agentA}$
      there exists $\actionStateT \in \actionStateS \actionAccessibilityAgent{\agentA}$ 
      such that $\pointedActionModel{\actionStateT} \bisimilar_{(n - 1)} \pointedActionModel[\prime]{\actionStateT[\prime]}$ 
  \end{definition}

  \begin{proposition}
      The relation $\bisimilar_n$ is an equivalence relation on action models.
  \end{proposition}

  \begin{proposition}
      Let $n \in \mathbb{N}$ and 
      let $\pointedActionModel{\actionStateS}, \pointedActionModel[\prime]{\actionStateS[\prime]} \in \classK$ be action models 
      such that $\pointedActionModel{\actionStateS} \bisimilar_n \pointedActionModel[\prime]{\actionStateS[\prime]}$. 
      If $m < n$ then $\pointedActionModel{\actionStateS} \bisimilar_m \pointedActionModel[\prime]{\actionStateS[\prime]}$. 
  \end{proposition}

  \begin{proposition}
      Let $n \in \mathbb{N}$,
      let $\pointedModel{\stateS}, \pointedModel[\prime]{\stateS[\prime]} \in \classK$ be Kripke models such that
      $\pointedModel{\stateS} \bisimilar_n \pointedModel[\prime]{\stateS[\prime]}$. 
      and let $\pointedActionModel{\actionStateS}, \pointedActionModel[\prime]{\actionStateS[\prime]} \in \classK$ be action models such that
      $\pointedActionModel{\actionStateS} \bisimilar_n \pointedActionModel[\prime]{\actionStateS[\prime]}$.
      Then
      $(\pointedModel{\stateS} \exec \pointedActionModel{\actionStateS}) \bisimilar_n (\pointedModel[\prime]{\stateS[\prime]} \exec \pointedActionModel[\prime]{\actionStateS[\prime]})$.
  \end{proposition}

  \begin{proposition}
      Let $\pointedActionModel{\actionStateS}, \pointedActionModel[\prime]{\actionStateS[\prime]} \in \classK$ be action models.
      Then $\pointedActionModel{\actionStateS} \bisimilar \pointedActionModel[\prime]{\actionStateS[\prime]}$
      if and only if for every $n \in \mathbb{N}$: 
      $\pointedActionModel{\actionStateS} \bisimilar_n \pointedActionModel[\prime]{\actionStateS[\prime]}$.
  \end{proposition}

  These results follow from similar reasoning to the results for $n$-bisimilarity of Kripke models.

  \begin{definition}[$\agentsB$-bisimilarity of action models]
      Let $\pointedActionModel{\actionStateS} = \pointedActionModelTuple{\actionStateS} \in \classK$ 
      and $\pointedActionModel[\prime]{\actionStateS[\prime]} = \pointedActionModelTuple[\prime]{\actionStateS[\prime]} \in \classK$
      be Kripke models. 
      We say that $\pointedActionModel{\actionStateS}$ is $\agentsB$-bisimilar
      to $\pointedActionModel[\prime]{\actionStateS[\prime]}$ 
      and write $\pointedActionModel{\actionStateS} \bisimilar_\agentsB \pointedActionModel[\prime]{\actionStateS[\prime]}$,
      if and only if
      for every $\agentB \in \agentsB$ the following conditions hold:

      \paragraph{atoms}
      For every $\atomP \in \atoms$: $\actionStateS \in \valuation(\atomP)$ if and only if $\actionStateS[\prime] \in \valuation[\prime](\atomP)$.

      \paragraph{forth-$\agentB$}
      For every $\actionStateT \in \actionStateS \actionAccessibilityAgent{\agentB}$ 
      there exists $\actionStateT[\prime] \in \actionStateS[\prime] \actionAccessibilityAgent[\prime]{\agentB}$
      such that $\pointedActionModel{\actionStateT} \bisimilar \pointedActionModel[\prime]{\actionStateT[\prime]}$.

      \paragraph{back-$\agentB$}
      For every $\actionStateT[\prime] \in \actionStateS[\prime] \actionAccessibilityAgent[\prime]{\agentB}$
      there exists $\actionStateT \in \actionStateS \actionAccessibilityAgent{\agentB}$ 
      such that $\pointedActionModel{\actionStateT} \bisimilar \pointedActionModel[\prime]{\actionStateT[\prime]}$ .
  \end{definition}

  \begin{proposition}
      Let $\pointedModel{\stateS}, \pointedModel[\prime]{\stateS[\prime]} \in \classK$ be Kripke models such that
      $\pointedModel{\stateS} \bisimilar_\agentsB \pointedModel[\prime]{\stateS[\prime]}$. 
      and let $\pointedActionModel{\actionStateS}, \pointedActionModel[\prime]{\actionStateS[\prime]} \in \classK$ be action models such that
      $\pointedActionModel{\actionStateS} \bisimilar_\agentsB \pointedActionModel[\prime]{\actionStateS[\prime]}$.
      Then
      $(\pointedModel{\stateS} \exec \pointedActionModel{\actionStateS}) \bisimilar_\agentsB (\pointedModel[\prime]{\stateS[\prime]} \exec \pointedActionModel[\prime]{\actionStateS[\prime]})$.
  \end{proposition}

  This result follows from similar reasoning to the results for $\agentsB$-bisimilarity of Kripke models.

  \begin{definition}[Axiomatisation \axiomAmlK]\label{aml-k-axioms}
  The axiomatisation \axiomAmlK{} is a substitution schema consisting of the
  rules and axioms of \axiomK{} along with the axioms:
  $$
  \begin{array}{rl}
      {\bf AS}    & \proves \allacts{\pointedActionModel{\actionStatesT} \exec \pointedActionModel[\prime]{\actionStatesT[\prime]}} \phi \iff \allacts{\pointedActionModel[\prime]{\actionStatesT[\prime]}} \allacts{\pointedActionModel{\actionStatesT}} \phi\\
      {\bf AU}    & \proves \allacts{\pointedActionModel{\actionStatesT}} \phi \iff \bigwedge_{\actionStateT \in \actionStatesT} \allacts{\pointedActionModel{\actionStateT}} \phi\\
      {\bf AP}    & \proves \allacts{\pointedActionModel{\actionStateT}} \atomP \iff (\actionPrecondition(\actionStateT) \implies \atomP) \text{ for $\atomP \in \atoms$}\\
      {\bf AN}    & \proves \allacts{\pointedActionModel{\actionStateT}} \neg \phi \iff (\actionPrecondition(\actionStateT) \implies \neg \allacts{\pointedActionModel{\actionStateT}} \phi)\\
      {\bf AC}    & \proves \allacts{\pointedActionModel{\actionStateT}} (\phi \land \psi) \iff (\allacts{\pointedActionModel{\actionStateT}} \phi \land \allacts{\pointedActionModel{\actionStateT}} \psi)\\ 
      {\bf AK}    & \proves \allacts{\pointedActionModel{\actionStateT}} \necessary[\agentA] \phi \iff (\actionPrecondition(\actionStateT) \implies \necessary[\agentA] \allacts{\pointedActionModel{\actionStateT \actionAccessibilityAgent{\agentA}}} \phi)\\
  \end{array}
  $$
  and the rule:
  $$
  \begin{array}{rl}
      {\bf NecA} & \text{From $\proves \phi$ infer $\proves \allacts{\pointedActionModel{\actionStatesT}} \phi$}
  \end{array}
  $$
  \end{definition}

  \begin{proposition}\label{aml-k-soundness-completeness}
      The axiomatisation \axiomAmlK{} is sound and complete for the logic \logicAmlK{}.
  \end{proposition}

  \begin{proposition}\label{aml-k-expressive-equivalence}
      The logic \logicAmlK{} is expressively equivalent to the logic \logicK{}.
  \end{proposition}

  These results are shown by Baltag, Moss and Solecki~\cite{baltag1998,baltag2005}. 
  We note that the completeness and expressive equivalence results follow from
  the fact that \axiomAmlK{} forms a set of reduction axioms which give a
  provably correct translation from \langAml{} to \lang{}.

  We note that the same results hold for the logics \logicAmlKFF{} and
  \logicAmlS{} if we extend \axiomAmlK{} with the additional axioms of
  \axiomKFF{} and \axiomS{} and restrict the language to only include
  $\classAM_\classKFF$ and $\classAM_\classS$ action models respectively, given
  the following results.

  \begin{proposition}\label{aml-kff-domain}
      $\model \in \classK$ and $\actionModel \in \classAM_\classKFF$ if and only if $\model \exec \actionModel \in \classKFF$.
  \end{proposition}

  \begin{proposition}\label{aml-s-domain}
      $\model \in \classK$ and $\actionModel \in \classAM_\classS$ if and only if $\model \exec \actionModel \in \classS$.
  \end{proposition}
  
  \begin{definition}[Simulation and refinement]\label{refinements}
      Let $\model, \model[\prime] \in \classK$ be Kripke models.
      A non-empty relation $\refinementRel \subseteq \states \times \states[\prime]$
      is a {\em simulation} if and only if it satisfies {\bf atoms}, {\bf forth-$\agentA$} for every $\agentA \in \agents$.
      If $(\stateS, \stateS[\prime]) \in \refinementRel$ then we call $\pointedModel[\prime]{\stateS[\prime]}$ a {\em simulation} of $\pointedModel{\stateS}$
      and call $\pointedModel{\stateS}$ a {\em refinement} of $\pointedModel[\prime]{\stateS[\prime]}$.
      We write $\pointedModel[\prime]{\stateS[\prime]} \simulation \pointedModel{\stateS}$
      or equivalently $\pointedModel{\stateS} \refinement \pointedModel[\prime]{\stateS[\prime]}$.
  \end{definition}

  \begin{proposition}
      The relation $\refinement$ is a preorder on Kripke models.
  \end{proposition}

  \begin{proposition}
      Let $\pointedModel{\stateS} \in \classK$ and $\pointedActionModel{\actionStateS} \in \classAM$.
      Then $\pointedModel{\stateS} \exec \pointedActionModel{\actionStateS} \refinement \pointedModel{\stateS}$
  \end{proposition}

  These results are shown by van Ditmarsch and French~\cite{vanditmarsch2009}.

  \begin{definition}[Semantics of arbitrary action model logic]
      Let \classC{} be a class of Kripke models 
      and let $\model \in \classC$ be a Kripke model.
      The interpretation of $\phi \in \langAaml$ in the logic \logicAamlC{}
      is the same as its interpretation in the action model logic \logicAmlC{} given in
      Definition~\ref{aml-semantics} with the additional inductive case:
      \begin{eqnarray*}
          \pointedModel{\stateS} \entails \allrefs \phi &\text{ iff }& \text{for every } \pointedModel[\prime]{\stateS[\prime]} \in \classC \text{ such that } \pointedModel[\prime]{\stateS[\prime]} \refinement \pointedModel{\stateS}: \pointedModel[\prime]{\stateS[\prime]} \entails \phi
      \end{eqnarray*}
  \end{definition}

  The semantics of arbitrary action model logic are given by Hales~\cite{hales2013},
  which are a combination of the semantics of action model logic of Baltag, Moss and Solecki~\cite{baltag1998,baltag2005}
  and the semantics of refinement modal logic of van Ditmarsch and French~\cite{vanditmarsch2009}.

  As noted earlier, the action model logics \logicAmlK{}, \logicAmlKFF{} and
  \logicAmlS{} are expressively equivalent to their underlying modal logics via
  a provably correct translation. Similarly it was shown by Bozzelli, et
  al.~\cite{bozzelli2012a} and Hales, French and Davies~\cite{hales2012} that
  the refinement modal logics \logicRmlK{}, \logicRmlKD{} and \logicS{} are
  expressively equivalent to their underlying modal logics, also via a provably
  correct translation. We note that the same result for \logicRmlKFF{} can be
  shown similarly to the result for \logicRmlKD{}. In axiomatising
  \logicAamlK{}, Hales~\cite{hales2013} simply noted that the rules and axioms
  of \axiomAmlK{} and \axiomRmlK{} are sound in \logicAamlK{} and that the
  provably correct translations for \logicAmlK{} and \logicRmlK{} can be
  simply combined to form a provably correct translation for \logicAamlK{}.
  We reproduce the axiomatisation for \logicAamlK{} here, and note that the same
  similar reasoning to~\cite{hales2013} gives sound and complete
  axiomatisations and provably correct translations for \logicAamlKFF{} and
  \logicAamlS{}, which we also list here.

  \begin{definition}[Disjunctive normal form]\label{dnf}
      A formula in {\em disjunctive normal form} is defined by the following abstract syntax:
      $$
      \phi :: \pi \land \bigwedge_{\agentB \in \agentsB} \covers_\agentA \Gamma_\agentA \mid \phi \lor \phi
      $$
      where $\pi \in \langP$, $\agentsB \subseteq \agents$ and for every $\agentB \in \agentsB$,
      $\Gamma_\agentB$ is a finite set of formulae in disjunctive normal form.
  \end{definition}

  \begin{proposition}
      Every formula of \lang{} is equivalent to a formula in disjunctive normal
      form under the semantics of \logicK{}.
  \end{proposition}

  This is shown by van Ditmarsch, French and Pinchinat~\cite{vanditmarsch2010}.

  \begin{definition}[Axiomatisation \axiomAamlK{}]
      The axiomatisation \axiomAamlK{} is a substitution schema consisting of
      the rules and axioms of \axiomAmlK{} along with the axioms:
      $$
      \begin{array}{rl}
          {\bf R}   & \allrefs (\phi \implies \psi) \implies (\allrefs \phi \implies \allrefs \psi)\\
          {\bf RP}  & \allrefs \pi \iff \pi \text{ where } \pi \in \langP\\
          {\bf RK}  & \somerefs \covers_\agentA \Gamma_\agentA \iff \bigwedge_{\gamma \in \Gamma_\agentA} \possible[\agentA] \somerefs \gamma\\
          {\bf RDist}  & \somerefs \bigwedge_{\agentA \in \agents} \covers_\agentA \Gamma_\agentA \iff \bigwedge_{\agentA \in \agents} \somerefs \covers_\agentA \Gamma_\agentA
      \end{array}
      $$
      and the rule:
      $$
      \begin{array}{rl}
          {\bf NecR} & \text{From $\proves \phi$ infer $\proves \allrefs \phi$}
      \end{array}
      $$
  \end{definition}
  
  The additional axioms for \axiomAamlK{} are the additional axioms from \axiomRmlK{} for refinement modal logic, given by Bozzelli, et al.~\cite{bozzelli2012a}.

  \begin{proposition}
      The axiomatisation \axiomAamlK{} is sound and complete for the logic \logicAamlK{}.
  \end{proposition}

  \begin{proposition}
      The logic \logicAamlK{} is expressively equivalent to the logic \logicK{}.
  \end{proposition}

  These results are shown by Hales~\cite{hales2013}.

  \begin{definition}[Alternating disjunctive normal form]\label{adnf}
      A formula in {\em $\agentA$-alternating disjunctive normal form} is defined by the following abstract syntax:
      $$
      \phi :: \pi \land \bigwedge_{\agentB \in \agentsB} \covers_\agentB \Gamma_\agentB \mid \phi \lor \phi
      $$
      where $\pi \in \langP$, $\agentsB \subseteq \agents \setminus \{\agentA\}$ and for every $\agentB \in \agentsB$,
      $\Gamma_\agentB$ is a finite set of formulae in $\agentB$-alternating disjunctive normal form.

      A formula in {\em alternating disjunctive normal form} is defined by the following abstract syntax:
      $$
      \phi :: \pi \land \bigwedge_{\agentB \in \agentsB} \covers_\agentA \Gamma_\agentA \mid \phi \lor \phi
      $$
      where $\pi \in \langP$, $\agentsB \subseteq \agents$ and for every $\agentB \in \agentsB$,
      $\Gamma_\agentB$ is a finite set of formulae in $\agentB$-alternating disjunctive normal form.
  \end{definition}
  
  The additional axioms for \axiomAamlKFF{} are adapted from the additional
  axioms from \axiomRmlKD{} for refinement doxastic logic, given by Hales, French
  and Davies~\cite{hales2012}. The axioms do not require that each
  $\Gamma_\agentA$ be non-empty, which is due to the lack of seriality in the
  setting of \classKFF{}.

  \begin{proposition}
      Every formula of \lang{} is equivalent to a formula in alternating
      disjunctive normal form under the semantics of \logicKFF{}.
  \end{proposition}

  This is shown by Hales, French and Davies~\cite{hales2012} for \logicKD{},
  however the same reasoning applies to \logicKFF{}.

  \begin{definition}[Axiomatisation \axiomAamlKFF{}]
      The axiomatisation \axiomAmlKFF{} is a substitution schema consisting of
      the rules and axioms of \axiomAmlKFF{} along with the rules and axioms
      {\bf R}, {\bf RP} and {\bf NecR} of \axiomAamlK{} and the axioms:
      $$
      \begin{array}{rl}
          {\bf RK45}  & \somerefs \covers_\agentA \Gamma_\agentA \iff \bigwedge_{\gamma \in \Gamma_\agentA} \possible[\agentA] \somerefs \gamma\\
          {\bf RDist}  & \somerefs \bigwedge_{\agentA \in \agents} \covers_\agentA \Gamma_\agentA \iff \bigwedge_{\agentA \in \agents} \somerefs \covers_\agentA \Gamma_\agentA
      \end{array}
      $$
      where for every $\agentA \in \agents$, $\Gamma_\agentA$ is a finite set
      of $a$-alternating disjunctive normal formulae.
  \end{definition}

  The additional axioms for \axiomAamlKFF{} are the additional axioms from \axiomRmlKFF{} for refinement epistemic logic, given by Hales, French and Davies~\cite{hales2012}.
  
  The additional axioms for \axiomAamlKFF{} are adapted from the additional
  axioms from \axiomRmlKD{} for refinement modal logic, given by Hales, French
  and Davies~\cite{hales2012}. The axioms do not require that each
  $\Gamma_\agentA$ be non-empty, which is due to the lack of seriality in the
  setting of \classKFF{}.

  \begin{proposition}
      The axiomatisation \axiomAamlKFF{} is sound and complete for the logic \logicAamlKFF{}.
  \end{proposition}

  \begin{proposition}
      The logic \logicAamlKFF{} is expressively equivalent to the logic \logicKFF{}.
  \end{proposition}

  These results follow from similar reasoning to the same results in \logicAamlK{}.

  \begin{definition}[Explicit formulae]\label{explicit}
      Let $\pi \in \langP$ be a propositional formula,
      let $\gamma^0 \in \lang$ be a modal formula
      and for every $\agentA \in \agents$ 
      let $\Gamma_\agentA \subseteq \lang$ 
      be a finite set of formulae such that $\gamma^0 \in \Gamma_\agentA$.
      Let $\Psi = \{\psi \leq \gamma \mid \agentA \in \agents, \gamma \in \Gamma_\agentA\}$
      be the set of subformulae of the formulae in each set $\Gamma_\agentA$.
      Finally let $\phi$ be a formula of the form
      $$
      \phi = \pi \land \gamma^0 \land \bigwedge_{\agentA \in \agents} \covers_\agentA \Gamma_\agentA
      $$
      Then $\phi$ is an {\em explicit formula} if and only if the following conditions hold:

      \begin{enumerate}
          \item For every $\agentA \in \agents$, $\gamma \in \Gamma_\agentA$, $\psi \in \Psi$: 
              either $\proves_\axiomS \gamma \implies \psi$ or $\proves_\axiomS \gamma \implies \neg \psi$.
          \item For every $\agentA \in \agents$, $\gamma \in \Gamma_\agentA$, $\necessary_\agentA \psi \in \Psi$: 
              $\proves_\axiomS \gamma \implies \necessary[\agentA] \psi$ if and only if 
              for every $\gamma' \in \Gamma_\agentA$: $\proves_\axiomS \gamma' \implies \psi$.
      \end{enumerate}
  \end{definition}

  \begin{proposition}
      Every formula of \lang{} is equivalent to a disjunction of explicit
      formulae under the semantics of \logicS{}.
  \end{proposition}

  This is shown by Hales, French and Davies~\cite{hales2012}.

  \begin{definition}[Axiomatisation \axiomAamlS{}]
      The axiomatisation \axiomAmlS{} is a substitution schema consisting of
      the rules and axioms of \axiomAmlS{} along with the rules and axioms
      {\bf R}, {\bf RP} and {\bf NecR} of \axiomAamlK{} and the axioms:
      $$
      \begin{array}{rl}
          {\bf RS5}  & \somerefs (\gamma^0 \land \covers_\agentA \Gamma_\agentA) \iff \somerefs \gamma^0 \land \bigwedge_{\gamma \in \Gamma_\agentA} \possible[\agentA] \somerefs \gamma\\
          {\bf RDist}  & \somerefs (\gamma^0 \land \bigwedge_{\agentA \in \agents} \covers_\agentA \Gamma_\agentA) \iff \bigwedge_{\agentA \in \agents} \somerefs (\gamma^0 \land \covers_\agentA \Gamma_\agentA)
      \end{array}
      $$
      where $\gamma^0 \land \bigwedge_{\agentA \in \agents} \somerefs \covers_\agentA \Gamma_\agentA$ is an explicit formula
      and for every $\agentA \in \agents$, $\gamma^0 \land \covers_\agentA \Gamma_\agentA$ is an explicit formula.
  \end{definition}

  \begin{proposition}
      The axiomatisation \axiomAamlS{} is sound and complete for the logic \logicAamlS{}.
  \end{proposition}

  \begin{proposition}
      The logic \logicAamlS{} is expressively equivalent to the logic \logicS{}.
  \end{proposition}

  These results follow from similar reasoning to the same results in \logicAamlK{}.

  \section{Syntax}\label{syntax}

  \begin{definition}[Language of arbitrary action formula logic]
      The language \langAafl{} of arbitrary action formula logic is inductively defined as:
      $$
          \phi ::= \atomP \mid 
                 \neg \phi \mid
                 (\phi \land \phi) \mid
                 \necessary[\agentA] \phi \mid
                 \allacts{\alpha} \phi \mid
                 \allrefs \phi
      $$
      where $\atomP \in \atoms$, $\agentA \in \agents$ and 
      $\alpha \in \langAaflAct{}$, and where the language \langAaflAct{} of 
      arbitrary action formulae is inductively as:
      $$
          \alpha ::= \test{\phi} \mid
                 \alpha \choice \alpha \mid
                 \alpha \compose \alpha \mid
                 \learns_\agentsB (\alpha, \alpha)
      $$
      where $\phi \in \langAafl{}$ and $\emptyset \subset \agentsB \subseteq \agents$.
  \end{definition}

  We use all of the standard abbreviations for arbitrary action model logic, in
  addition to the abbreviations
  $\learns_\agentsB \alpha ::= \learns_\agentsB (\alpha, \alpha)$ and 
  $\learns_\agentA (\alpha, \beta) ::= \learns_{\{\agentA\}} (\alpha, \beta)$.

  We denote non-deterministic choice ($\choice$) over a finite set of action
  formula $\Delta \subseteq \langAaflAct$ by $\bigchoice \Delta$ 
  and we denote sequential execution ($\compose$) of a finite, non-empty
  sequence of action formulae $(\alpha_i)_{i=0}^{n} \in \mathbb{N}^\langAaflAct$ 
  by $\bigcompose (\alpha_i)_{i=0}^{n}$ and define them in the obvious way.

  We refer to the languages \langAfl{} of action formula logic and \langAflAct{} of action formulae, which are \langAafl{} and \langAaflAct{} respectively, both without the $\allrefs$ operator, 

  As in the action model logic~\cite{baltag2005}, the intended meaning of the
  operator $\allacts{\alpha} \phi$ is that ``$\phi$ is true in the result of
  any successful execution of the action $\alpha$''. In the following section
  we define the semantics of the action formula logic in terms of action model
  execution. For each setting of \classK{}, \classKFF{} and \classS{} we
  provide a function $\tau_\classC : \langAflAct \to \classAM$ of translating
  action formulae from \langAflAct{} into action models. The result of
  executing an action $\alpha \in \langAflAct{}$ is determined by translating
  $\alpha$ into an action model $\tau_\classC(\alpha) \in \classAM_\classC$, and then executing
  the action model in the usual way.

  In each setting we have attempted to define the translation from action
  formulae into action models in such a way that the action formulae carry an
  intuitive description of the action that is performed by the corresponding
  action model. We call the $\test{}$ operator the test operator, and describe
  the action $\test{\phi}$ as a test for $\phi$. A test is intended to restrict
  the states in which an action can successfully execute to states where the
  condition $\phi$ is true initially, but otherwise leaves the state unchanged. 
  We call the $\choice$ operator the non-deterministic choice operator, and
  describe the action $\alpha \choice \beta$ as a non-deterministic choice
  between $\alpha$ and $\beta$. We call the $\compose$ operator the sequential
  execution operator, and describe the action $\alpha \compose \beta$ as an
  execution of $\alpha$ followed by $\beta$.  Finally we call $\learns_\agentsB$
  the learning operator, and describe the action $\learns_\agentsB (\alpha,
  \beta)$ as the agents in $\agentsB$ learning that the actions $\alpha$ or
  $\beta$ occurred.  
  This action is intended to result in the agents $\agentsB$
  knowing or believing what would be true if $\alpha$ or $\beta$ were executed.
  For example, if a consequence of executing $\alpha$ is that $\phi$ is true in
  the result, then the intention is that a consequence of executing
  $\learns_\agentA (\alpha, \alpha)$ is that $\knows_\agentA \phi$ is true in
  the result.  As we will see, this property is generally true in \logicAflK{},
  however due to the extra frame conditions of \classKFF{} and \classS{} it is
  only true for some formulae $\phi$ in \logicAflKFF{} and \logicAflS{}.

  \begin{example}\label{grant-example-formula}
      If $\atomP$ stands for the proposition ``the grant application was
      successful'' then the action described in Example~\ref{grant-example}
      might be written in the form of an action formula as:
      \begin{align*}
          \alpha = &\learns_{Ed} (\test{\atomP}) \compose\\
          &\learns_{James} (\learns_{Ed} \test{\atomP} \choice \learns_{Ed} \test{\neg \atomP} \choice \learns_{Tim} \test{\atomP} \choice \learns_{Tim} \test{\neg \atomP}) \compose\\
          &\learns_{Tim} ((\test{\neg \atomP} \compose \learns_{James} \test{\neg \atomP}) \choice \test{\top})
      \end{align*}
  \end{example}

  \section{Semantics}\label{semantics}

  We now define the semantics of arbitrary action formula logic. As mentioned
  earlier, the semantics are defined by translating action formulae into action
  models.  The translation used varies in each class of \classK{}, \classKFF{}
  and \classS{} that we work in, according to the frame conditions in each
  class.  Therefore our semantics are parameterised by a function
  $\tau_\classC: \langAflAct \to \classAM$ that will vary according to the
  class of Kripke models.

  \begin{definition}[Semantics of arbitrary action formula logic]
      Let \classC{} be a class of Kripke models, let $\tau_\classC :
      \langAflAct \to \classAM$ be a function from action formulae to
      multi-pointed action models, and let $\model = \modelTuple \in \classC$
      be a Kripke model.

      Then the interpretation of $\phi \in \langAafl$ in the logic
      $\logicAaflC$ is the same as its interpretation in modal logic given in
      Definition~\ref{ml-semantics}, with the additional inductive cases:
      \begin{eqnarray*}
          \pointedModel{\stateS} \entails \allacts{\alpha} \phi &\text{ iff }& \pointedModel{\stateS} \exec \tau_\classC(\alpha) \in \classC \text{ implies } \pointedModel{\stateS} \exec \tau_\classC(\alpha) \entails \phi\\
          \pointedModel{\stateS} \entails \allrefs \phi &\text{ iff }& \text{for every } \pointedModel[\prime]{\stateS[\prime]} \in \classC \text{ such that } \pointedModel[\prime]{\stateS[\prime]} \refinement \pointedModel{\stateS}: \pointedModel[\prime]{\stateS[\prime]} \entails \phi
      \end{eqnarray*}
      where action model execution $\exec$ is as defined in Definition~\ref{aml-semantics} 
      and the refinement relation is defined in Definition~\ref{refinements}.
  \end{definition}

  We note that the semantics of arbitrary action formula logic \logicAaflC{} are very similar
  to the semantics of arbitrary action model logic \logicAamlC{}~\cite{hales2013}. 
  We generalise the semantics to the classes of \classK{}, \classKFF{} and
  \classS{} by introducing the parameterised class \classC{} and restricting
  successful updates to those that result in \classC{} models as in the
  approach of Balbiani, et al~\cite{balbiani2012}.
  The difference is that as actions are specified in \langAafl{} formulae as
  action formulae, then the semantics must first translate the action formulae
  into action models before performing action model execution. As such there is
  a semantically correct translation from \langAafl{} formulae to \langAaml{}
  formulae (by replacing occurrences of $\alpha$ with $\tau_\classC(\alpha)$),
  and any validities, axioms or results from arbitrary action model logic also
  apply in this setting if the language is restricted to action models that are
  defineable by action formulae. Therefore for the current section and the
  following sections concering the axiomatisations
  (Section~\ref{axiomatisation}) and correspondence results
  (Section~\ref{correspondence}), we will deal only with the action formula
  logic, rather than the full arbitrary action formula logic, focussing on the
  differences and correspondences between action formulae and action models,
  rather than getting distracted by the refinement quantifiers which behave
  identically between each logic. We return to the full arbitrary action
  formula logic in Section~\ref{synthesis} for the synthesis results.

  We give the following general result.

  \begin{proposition}
      Let $\classC$ be a class of Kripke models. For every $\phi \in \langAafl$
      there exists $\phi' \in \langAaml$ 
      such that for every $\pointedModel{\statesT} \in \classC$: 
      $\pointedModel{\statesT} \entails_\logicAaflC \phi$ if and only if
      $\pointedModel{\statesT} \entails_\logicAamlC \phi'$.
  \end{proposition}

  In the following subsections we will give definitions for $\tau_\classK$,
  $\tau_\classKFF$ and $\tau_\classS$. These functions vary according to the
  class of Kripke models being used. When the class is clear from context, then
  we will simply write $\tau$ instead of $\tau_\classC$.

  We begin by giving a definition of $\tau$ for translating actions involving
  non-deterministic choice and sequential execution. These definitions are
  common to all of the settings we are working in.

  \begin{definition}[Non-deterministic choice]\label{afl-choice}
      Let $\classC \in \{\classK, \classKFF, \classS\}$
      and let $\alpha, \beta \in \langAflAct$ 
      where $\tau_\classC(\alpha) = \pointedActionModel[\alpha]{\actionStatesT[\alpha]} = \pointedActionModelTuple[\alpha]{\actionStatesT[\alpha]}$
      and $\tau_\classC(\beta) = \pointedActionModel[\beta]{\actionStatesT[\beta]} = \pointedActionModelTuple[\beta]{\actionStatesT[\beta]}$
      such that $\actionStates[\alpha]$ and $\actionStates[\beta]$ are disjoint.
      We define $\tau_\classC(\alpha \choice \beta) = \pointedActionModel{\actionStatesT} = \pointedActionModelTuple{\actionStatesT}$ where:
      \begin{eqnarray*}
          \actionStates &=& \actionStates[\alpha] \cup \actionStates[\beta]\\
          \actionAccessibilityAgent{\agentA} &=& \actionAccessibilityAgent[\alpha]{\agentA} \cup \actionAccessibilityAgent[\beta]{\agentA} \text{ for } \agentA \in \agents\\
          \actionPrecondition &=& \actionPrecondition[\alpha] \cup \actionPrecondition[\beta]\\
          \actionStatesT &=& \actionStatesT[\alpha] \cup \actionStatesT[\beta]
      \end{eqnarray*}
  \end{definition}

  \begin{definition}[Sequential execution]\label{afl-sequential}
      Let $\classC \in \{\classK, \classKFF, \classS\}$,
      and let $\alpha, \beta \in \langAflAct$ where 
      $\tau_\classC(\alpha) = \pointedActionModel[\alpha]{\actionStatesT[\alpha]} = \pointedActionModelTuple[\alpha]{\actionStatesT[\alpha]}$ and
      $\tau_\classC(\beta) = \pointedActionModel[\beta]{\actionStatesT[\beta]} = \pointedActionModelTuple[\beta]{\actionStatesT[\beta]}$.
      We define $\tau_\classC(\alpha \compose \beta) = \pointedActionModel[\alpha]{\actionStatesT[\alpha]} \exec \pointedActionModel[\beta]{\actionStatesT[\beta]}$.
  \end{definition}

  We give some properties of non-deterministic choice and sequential execution
  of action formulae.

  \begin{proposition}\label{afl-choice-sequential-validities}
      Let $\alpha, \beta, \gamma \in \langAflAct$ and $\phi \in \langAfl$. Then the following are valid in \logicAflK{}, \logicAflKFF{} and \logicAflS{}:
      \begin{eqnarray*}
          &&\entails \allacts{\alpha \choice \beta} \phi \iff (\allacts{\alpha} \phi \land \allacts{\beta} \phi) \label{afl-axiom-choice}\\
          &&\entails \allacts{\alpha \compose \beta} \phi \iff \allacts{\alpha} \allacts{\beta} \phi \label{afl-axiom-sequential}\\
          &&\entails \allacts{\alpha \choice \alpha} \phi \iff \allacts{\alpha} \phi\\
          &&\entails \allacts{\alpha \choice \beta} \phi \iff \allacts{\beta \choice \alpha} \phi\\
          &&\entails \allacts{(\alpha \choice \beta) \choice \gamma} \phi \iff \allacts{\alpha \choice (\beta \choice \gamma)} \phi\\
          &&\entails \allacts{(\alpha \compose \beta) \compose \gamma} \phi \iff \allacts{\alpha \compose (\beta \compose \gamma)} \phi\\
          &&\entails \allacts{(\alpha \choice \beta) \compose \gamma} \phi \iff \allacts{(\alpha \compose \gamma) \choice (\beta \compose \gamma)} \phi
      \end{eqnarray*}
  \end{proposition}

  These validities follow trivially from the semantics of \logicAaflC{} and Definitions~\ref{afl-choice} and~\ref{afl-sequential}.

  In the following subsections we give definitions of $\tau_\classC$ for
  translating action formulae involving tests and learning in the settings of
  \classK{}, \classKFF{} and \classS{}. We note that in each subsection the
  constructions of action models used to define tests and learning closely
  resemble the constructions of refinements used to show the soundness of
  axioms in refinement modal logic~\cite{bozzelli2012a,hales2012}.

  \subsection{\classK{}}

  \begin{definition}[Test]\label{afl-k-test}
      Let $\phi \in \langAfl$. 
      We define $\tau(\test{\phi}) = \pointedActionModel{\actionStatesT} = \pointedActionModelTuple{\actionStatesT}$ where:
      \begin{eqnarray*}
          \actionStates &=& \{\actionStateTest, \actionStateSkip\}\\
          \actionAccessibilityAgent{\agentA} &=& \{(\actionStateTest, \actionStateSkip), (\actionStateSkip, \actionStateSkip)\} \text{ for } \agentA \in \agents\\
          \actionPrecondition &=& \{(\actionStateTest, \phi), (\actionStateSkip, \top)\}\\
          \actionStatesT &=& \{\actionStateTest\}
      \end{eqnarray*}
  \end{definition}

  \begin{definition}[Learning]\label{afl-k-learning}
      Let $\alpha \in \langAflAct$ where 
      $\tau(\alpha) = \pointedActionModel[\alpha]{\actionStatesT[\alpha]} = \pointedActionModelTuple[\alpha]{\actionStatesT[\alpha]}$.
      Let $\actionStateTest$ and $\actionStateSkip$ be new states not appearing in $\actionStates[\alpha]$.
      We define $\tau(\learns_\agentsB (\alpha, \alpha)) = \pointedActionModel{\actionStatesT} = \pointedActionModelTuple{\actionStatesT}$ where:
      \begin{eqnarray*}
          \actionStates &=& \actionStates[\alpha] \cup \{\actionStateTest, \actionStateSkip\}\\
          \actionAccessibilityAgent{\agentA} &=& \actionAccessibilityAgent[\alpha]{\agentA} \cup \{(\actionStateSkip, \actionStateSkip)\} \cup \{(\actionStateTest, \actionStateT[\alpha]) \mid \actionStateT[\alpha] \in \actionStatesT[\alpha]\} \text{ for } \agentA \in \agentsB\\
          \actionAccessibilityAgent{\agentA} &=& \actionAccessibilityAgent[\alpha]{\agentA} \cup \{(\actionStateTest, \actionStateSkip), (\actionStateSkip, \actionStateSkip)\} \text{ for } \agentA \notin \agentsB\\
          \actionPrecondition &=& \actionPrecondition[\alpha] \cup \{(\actionStateTest, \top), (\actionStateSkip, \top)\}\\
          \actionStatesT &=& \{\actionStateTest\}
      \end{eqnarray*}

      We define $\tau(\learns_\agentsB (\alpha, \beta)) = \tau(\learns_\agentsB (\alpha \choice \beta, \alpha \choice \beta))$.
  \end{definition}

  We note that the syntax of action formula logic defines the learning operator
  as a binary operator that can be applied to two different action formulae,
  however in the setting of \classK{} and \classKFF{} we only give a direct
  definition of $\tau$ for actions of the form $\learns_\agentsB (\alpha, \alpha)$
  and define the more general case in terms of this. Intuitively 
  $\learns_\agentsB (\alpha, \beta)$ is intended to represent an action where
  the agents in $\agentsB$ learn that $\alpha$ or $\beta$ have occurred (i.e.
  that $\alpha \choice \beta$ has occurred). The setting of \classS{}
  corresponds to a notion of {\em knowledge}, where anything that an agent {\em
  knows} must be true, and therefore anything that an agent {\em learns} must
  also be true. So in an action where agents learn that $\alpha$ or $\beta$
  have occurred, one of those actions must have actually occurred. Therefore in
  \classS{} we describe the action $\learns_\agentsB (\alpha, \beta)$ as the
  agents in $\agentsB$ learning that $\alpha$ or $\beta$ have occurred, when in
  reality $\alpha$ has actually occurred. On the other hand, the settings of
  \classK{} and \classKFF{} correspond more closely to a notion of {\em
  belief}, where there is no requirement that what an agent {\em believes} is
  true. So in an action where agents learn that $\alpha$ or $\beta$ have
  occurred, neither of these actions must actually have occurred. Therefore in
  the settings of \classK{} and \classKFF{} we make no distinction between
  $\alpha$ and $\beta$ in a description of the action $\learns_\agentsB
  (\alpha, \beta)$, hence the definition of $\tau$ given in these settings.

  \subsection{\classKFF{}}

  \begin{definition}[Test]\label{afl-kff-test}
      Let $\phi \in \langAfl$. We define $\tau(\test{\phi})$ as in
      Definition~\ref{afl-k-test} for \classK{}.
  \end{definition}

  \begin{definition}[Learning]\label{afl-kff-learning}
      Let $\alpha \in \langAflAct$ where 
      $\tau(\alpha) = \pointedActionModel[\alpha]{\actionStatesT[\alpha]} = \pointedActionModelTuple[\alpha]{\actionStatesT[\alpha]}$.
      Let $\actionStateTest$ and $\actionStateSkip$ be new states not appearing in $\actionStates[\alpha]$.
      For every $\actionStateT[\alpha] \in \actionStatesT[\alpha]$ let $\proxyStateT[\alpha]$ be a new state not appearing in $\actionStates[\alpha]$. 
      We call each $\proxyStateT[\alpha]$ a {\em proxy state} for $\actionStateT[\alpha]$.
      We define $\tau(\learns_\agentsB (\alpha, \alpha)) = \pointedActionModel{\actionStatesT} = \pointedActionModelTuple{\actionStatesT}$ where:
      \begin{eqnarray*}
          \actionStates &=& \actionStates[\alpha] \cup \{\actionStateTest, \actionStateSkip\} \cup \{\proxyStateT[\alpha] \mid \actionStateT[\alpha] \in \actionStatesT[\alpha]\}\\
          \actionAccessibilityAgent{\agentA} &=& \actionAccessibilityAgent[\alpha]{\agentA} \cup \{(\actionStateSkip, \actionStateSkip)\} \cup \{(\actionStateTest, \proxyStateT[\alpha]) \mid \actionStateT[\alpha] \in \actionStatesT[\alpha]\} \cup \\&& \quad \{(\proxyStateT[\alpha], \proxyStateU[\alpha]) \mid \actionStateT[\alpha], \actionStateU[\alpha] \in \actionStatesT[\alpha]\} \text{ for } \agentA \in \agentsB\\
          \actionAccessibilityAgent{\agentA} &=& \actionAccessibilityAgent[\alpha]{\agentA} \cup \{(\actionStateTest, \actionStateSkip), (\actionStateSkip, \actionStateSkip)\} \cup\\&& \quad \{(\proxyStateT[\alpha], \actionStateU[\alpha]) \mid \actionStateT[\alpha] \in \actionStatesT[\alpha], \actionStateU[\alpha] \in \actionStateT[\alpha] \actionAccessibilityAgent[\alpha]{\agentA} \} \text{ for } \agentA \notin \agentsB\\
          \actionPrecondition &=& \actionPrecondition[\alpha] \cup \{(\actionStateTest, \top), (\actionStateSkip, \top)\} \cup \{(\proxyStateT[\alpha], \actionPrecondition[\alpha](\actionStateT[\alpha])) \mid \actionStateT[\alpha] \in \actionStatesT[\alpha]\}\\
          \actionStatesT &=& \{\actionStateTest\}
      \end{eqnarray*}

      As in Definition~\ref{afl-k-learning}, we define $\tau(\learns_\agentsB (\alpha, \beta)) = \tau(\learns_\agentsB (\alpha \choice \beta, \alpha \choice \beta))$.
  \end{definition}

  \begin{lemma}\label{afl-kff-structure}
      Let $\alpha \in \langAflAct$. Then $\tau(\alpha) \in \classAM_\classKFF$.
  \end{lemma}

  \begin{lemma}\label{afl-kff-exec}
      Let $\alpha \in \langAflAct$ and 
      let $\pointedModel{\statesT} \in \classKFF$.
      Then $\pointedModel{\statesT} \exec \tau(\alpha) \in \classKFF$.
  \end{lemma}

  We note that the definition for $\tau$ given here varies considerably from
  the definition given in the setting of \classK{} due to the presence of the
  proxy states.  The proxy states are introduced due to the additional frame
  constraints in \classKFF{} and the desire that the action models constructed
  by $\tau$ be $\classAM_\classKFF$ action models. In constructing
  $\tau(\learns_\agentsB \alpha)$ we wish to construct an action model with a
  root state whose $\agentsB$-successors are the root states of $\tau(\alpha)$,
  so that the result of executing the action $\learns_\agentsB \alpha$ is that
  the agents $\agentsB$ believe that the action $\alpha$ has occurred. However
  in order for this construction to result in a $\classAM_\classKFF$ action
  model, we must take the transitive, Euclidean closure of the
  $\agentsB$-successors of the root state. If we were to perform a construction
  similar to that used in the setting of \classK{} where proxy states are not
  used, then this would mean that the for every $\agentB \in \agentsB$, the
  $\agentB$-successors of the root state would include all of the
  $\agentB$-successors of the root states, and not just the root states
  themselves. To show why this is not desireable, consider the simple example
  of the action $\learns_\agentA \test{\phi}$. The intention is that this
  action represents a private announcement to $\agentA$ that $\phi$ is true, as
  it is in the setting of \classK{}. Without using proxy states, if we wanted to
  include the state $\actionStateTest$ in the $\agentA$-successors of the root
  state of $\tau(\alpha)$ then in order to construct a $\classAM_\classKFF$
  action model we would need to take the transitive, Euclidean closure of the
  $\agentA$-successors of $\actionStateTest$. As $\actionStateSkip$ is an
  $\agentA$-successor of $\actionStateTest$ in the action $\test{\phi}$, then
  this would mean that $\agentA$ would not be able to distinguish between the
  actions states $\actionStateTest$ and $\actionStateSkip$ and so the result of
  executing $\tau(\alpha)$ would be that $\agentA$ learns nothing. With the
  construction provided, the action $\learns_\agentA \test{\phi}$ gives the
  desired result that $\agentA$ learns that $\phi$ is true.

  We also note that the results presented in this paper for \classKFF{} can be
  extended to \classKD{} by modifying Definition~\ref{afl-kff-learning} so that
  $\actionPrecondition(\actionStateTest) = \bigwedge_{\agentA \in \agentsB}
  \bigvee_{\actionStateT[\alpha] \in \actionStatesT[\alpha]} \possible[\agentA]
  \actionPrecondition[\alpha](\actionStateT[\alpha])$, which guarantees that
  the result of successfully executing an action formula has the seriality
  property of \classKD{}.

  \subsection{\classS{}}

  \begin{definition}[Test]\label{afl-s-test}
      Let $\phi \in \langAfl$. 
      We define $\tau(\test{\phi}) = \pointedActionModel{\actionStatesT} = \pointedActionModelTuple{\actionStatesT}$ where:
      \begin{eqnarray*}
          \actionStates &=& \{\actionStateTest, \actionStateSkip\}\\
          \actionAccessibilityAgent{\agentA} &=& \actionStates^2 \text{ for } \agentA \in \agents\\
          \actionPrecondition &=& \{(\actionStateTest, \phi), (\actionStateSkip, \top)\}\\
          \actionStatesT &=& \{\actionStateTest\}
      \end{eqnarray*}
  \end{definition}

  \begin{definition}[Learning]\label{afl-s-learning}
      Let $\alpha, \beta \in \langAflAct$ where 
      $\tau(\alpha) = \pointedActionModel[\alpha]{\actionStatesT[\alpha]} = \pointedActionModelTuple[\alpha]{\actionStatesT[\alpha]}$ and
      $\tau(\beta) = \pointedActionModel[\beta]{\actionStatesT[\beta]} = \pointedActionModelTuple[\beta]{\actionStatesT[\beta]}$.
      For every $\actionStateT \in \actionStatesT[\alpha] \cup \actionStatesT[\beta]$ let $\proxyStateT$ be a new state not appearing in $\actionStates[\alpha] \cup \actionStates[\beta]$.
      We define $\tau(\learns_\agentsB (\alpha, \beta)) = \pointedActionModel{\actionStatesT} = \pointedActionModelTuple{\actionStatesT}$ where:
      \begin{eqnarray*}
          \actionStates &=& \actionStates[\alpha] \cup \actionStates[\beta] \cup \{\proxyStateT \mid \actionStateT \in \actionStatesT[\alpha] \cup \actionStatesT[\beta]\}\\
          \actionAccessibilityAgent{\agentA} &=& \actionAccessibilityAgent[\alpha]{\agentA} \cup \actionAccessibilityAgent[\beta]{\agentA} \cup \{(\proxyStateT, \proxyStateU) \mid \actionStateT, \actionStateU \in \actionStatesT[\alpha] \cup \actionStatesT[\beta]\} \text{ for } \agentA \in \agentsB\\
          \actionAccessibilityAgent{\agentA} &=& \actionAccessibilityAgent[\alpha]{\agentA} \cup \actionAccessibilityAgent[\beta]{\agentA} \cup \bigcup_{\actionStateT \in \actionStatesT[\alpha] \cup \actionStatesT[\beta]} (\{\proxyStateT\} \cup \actionStateT (\actionAccessibilityAgent[\alpha]{\agentA} \cup \actionAccessibilityAgent[\beta]{\agentA}))^2 \text{ for } \agentA \notin \agentsB\\
          \actionPrecondition &=& \actionPrecondition[\alpha] \cup \actionPrecondition[\beta] \cup \{(\proxyStateT, (\actionPrecondition[\alpha] \cup \actionPrecondition[\beta])(\actionStateT)) \mid \actionStateT \in \actionStatesT[\alpha] \cup \actionStatesT[\beta]\}\\
          \actionStatesT &=& \{\proxyStateT \mid \actionStateT \in \actionStatesT[\alpha]\}
      \end{eqnarray*}
  \end{definition}

  \begin{lemma}\label{afl-s-structure}
      Let $\alpha \in \langAflAct$. Then $\tau(\alpha) \in \classAM_\classS$.
  \end{lemma}

  \begin{lemma}\label{afl-s-exec}
      Let $\alpha \in \langAflAct$ and 
      let $\pointedModel{\statesT} \in \classS$.
      Then $\pointedModel{\statesT} \exec \tau(\alpha) \in \classS$.
  \end{lemma}

  We note that as in the setting of \classKFF{} the definition of $\tau$ uses
  proxy states to construct action models from learning operators. However
  unlike in the settings of \classK{} and \classKFF{} this construction does
  not introduce the new states $\actionStateTest$ and $\actionStateSkip$. As
  discussed earlier this is because in the setting of \classS{}, in an action
  where agents learn that $\alpha$ or $\beta$ have occurred, one of those
  actions must have actually occurred. Unlike in the settings of \classK{} and
  \classKFF{} we have distinguished between the actions $\alpha$ and $\beta$,
  designating that $\alpha$ is the action that has actually occurred. We also
  note that the definition of $\tau$ for test operators is different from that
  used in \classK{} and \classKFF{}, simply to account for the additional frame
  constraints of \classS{}.

  \section{Axiomatisation}\label{axiomatisation}

  In the following subsections we give sound and complete axiomatisations for
  the action formulae logic in the settings of \classK{} and \classKFF{}. In
  the setting of \classS{} we provide a sound but not complete axiomatisation,
  and comment on the difficulty of giving a complete axiomatisation and the
  possible alternatives.
  We note that axiomatisations for arbitrary action formula logic in these
  settings can be derived trivially from these axiomatisations by adding the
  additional axioms and rules from refinement modal logic.

  \subsection{\classK{}}

  \begin{definition}[Axiomatisation \axiomAflK{}]\label{afl-k-axioms}
      The axiomatisation \axiomAflK{} is a substitution schema consisting of the
      rules and axioms of \axiomK{} along with the axioms:
      $$
      \begin{array}{rl}
          {\bf LT} & \proves \allacts{\test{\phi}} \psi \iff (\phi \implies \psi) \text{ for } \psi \in \lang\\
          {\bf LU} & \proves \allacts{\alpha \choice \beta} \phi \iff (\allacts{\alpha} \phi \land \allacts{\beta} \phi)\\
          {\bf LS} & \proves \allacts{\alpha \compose \beta} \phi \iff \allacts{\alpha} \allacts{\beta} \phi\\
          {\bf LP} & \proves \allacts{\learns_\agentsB (\alpha, \beta)} \atomP \iff \atomP\\
          {\bf LN} & \proves \allacts{\learns_\agentsB (\alpha, \beta)} \neg \phi \iff \neg \allacts{\learns_\agentsB (\alpha, \beta)} \phi\\
          {\bf LC} & \proves \allacts{\learns_\agentsB (\alpha, \beta)} (\phi \land \psi) \iff (\allacts{\learns_\agentsB (\alpha, \beta)} \phi \land \allacts{\learns_\agentsB (\alpha, \beta)} \psi)\\
          {\bf LK1} & \proves \allacts{\learns_\agentsB (\alpha, \beta)} \necessary[\agentA] \phi \iff \necessary[\agentA] \allacts{\alpha \choice \beta} \phi \text{ for } \agentA \in \agentsB\\
          {\bf LK2} & \proves \allacts{\learns_\agentsB (\alpha, \beta)} \necessary[\agentA] \phi \iff \necessary[\agentA] \phi \text{ for } \agentA \notin \agentsB
      \end{array}
      $$
      and the rule:
      $$
      \begin{array}{rl}
          {\bf NecL} & \text{From $\proves \phi$ infer $\proves \allacts{\alpha} \phi$}
      \end{array}
      $$
  \end{definition}

  \begin{proposition}\label{afl-k-axioms-soundness}
      The axiomatisation \axiomAflK{} is sound in the logic \logicAmlK{}.
  \end{proposition}

  \begin{proof}
      {\bf LT} follows from applying the reduction axioms of \axiomAmlK{}
      inductively to $\allacts{\test{\phi}} \psi$.

      {\bf LU} and {\bf LS} follow from Proposition~\ref{afl-choice-sequential-validities}.

      Let $\tau(\learns_\agentB (\alpha, \beta)) = \pointedActionModel{\actionStateS} = \pointedActionModelTuple{\actionStateS}$.
      {\bf LP}, {\bf LN} and {\bf LC} follow trivially from the \axiomAmlK{}
      axioms {\bf AP}, {\bf AN} and {\bf AC} respectively, noting from
      Definition~\ref{afl-k-learning} that $\actionPrecondition(\actionStateS) = \top$.
      {\bf LK1} follows trivially from the \axiomAmlK{} axiom {\bf AK},
      noting from Definition~\ref{afl-k-learning} that as $\agentA \in \agents$ then 
      $\pointedActionModel{\actionStateS \actionAccessibilityAgent{\agentA}} \bisimilar \tau(\alpha \choice \beta)$.
      {\bf NecL} follows trivially from the \axiomAmlK{} rule {\bf NecA}.
      {\bf LK2} follows trivially from the \axiomAmlK{} axiom {\bf AK},
      noting from Definition~\ref{afl-k-learning} that as $\agentA \notin \agents$ then 
      $\pointedActionModel{\actionStateS \actionAccessibilityAgent{\agentA}} \bisimilar \tau(\test{\top})$.
  \end{proof}

  \begin{proposition}\label{afl-k-axioms-completeness}
      The axiomatisation \axiomAflK{} is complete for the logic \logicAmlK{}.
  \end{proposition}

  We note that the axiomatisation \axiomAflK{} forms a set of reduction axioms
  that gives a provably correct translation from \langAfl{} to \lang{}.

  \begin{example}\label{grant-example-derivation}
      We give an example derivation that the action formula $\alpha$ given in
      Example~\ref{grant-example-formula} does indeed satisfy (part of) the
      epistemic goal stated in Example~\ref{grant-example}.
      \begin{eqnarray}
          &\proves& \allacts{\test{\atomP}} \atomP \iff (\atomP \implies \atomP)\label{grant-example-derivation-1}\\
          &\proves& \allacts{\test{\atomP}} \atomP\label{grant-example-derivation-2}\\
          &\proves& \necessary[Ed] \allacts{\test{\atomP}} \atomP\label{grant-example-derivation-3}\\
          &\proves& \allacts{\learns_{Ed} \test{\atomP}} \necessary[Ed] \atomP\label{grant-example-derivation-4}
      \end{eqnarray}
      (\ref{grant-example-derivation-1}) follows from {\bf LT},
      (\ref{grant-example-derivation-3}) follows from {\bf NecK} and
      (\ref{grant-example-derivation-4}) follows from {\bf LK1}.

      Similarly we have 
      \begin{eqnarray*}
          &\proves& \allacts{\learns_{Ed} \test{\neg \atomP}} \necessary[Ed] \neg \atomP\\
          &\proves& \allacts{\learns_{Tim} \test{\atomP}} \necessary[Tim] \atomP\\
          &\proves& \allacts{\learns_{Tim} \test{\neg \atomP}} \necessary[Tim] \neg \atomP
      \end{eqnarray*}

      Let $\phi = \necessary[Ed] \atomP \lor  \necessary[Ed] \neg \atomP \lor \necessary[Tim] \atomP \lor \necessary[Tim] \neg \atomP$. Then:
      \begin{eqnarray}
          &\proves& \allacts{\learns_{Ed} \test{\atomP} \choice \learns_{Ed} \test{\neg \atomP} \choice \learns_{Tim} \test{\atomP} \choice \learns_{Tim} \test{\neg \atomP}} \phi\label{grant-example-derivation-5}\\
          &\proves& \necessary[James] \allacts{\learns_{Ed} \test{\atomP} \choice \learns_{Ed} \test{\neg \atomP} \choice \learns_{Tim} \test{\atomP} \choice \learns_{Tim} \test{\neg \atomP}} \phi\label{grant-example-derivation-6}\\
          &\proves& \allacts{\learns_{James} (\learns_{Ed} \test{\atomP} \choice \learns_{Ed} \test{\neg \atomP} \choice \learns_{Tim} \test{\atomP} \choice \learns_{Tim} \test{\neg \atomP})} \necessary[James] \phi\label{grant-example-derivation-7}\\
          &\proves& \allacts{\alpha} \necessary[James] \phi\label{grant-example-derivation-8}
      \end{eqnarray}
      (\ref{grant-example-derivation-5}) follows from {\bf LU},
      (\ref{grant-example-derivation-6}) follows from {\bf NecK} and
      (\ref{grant-example-derivation-7}) follows from {\bf LK1}.
      (\ref{grant-example-derivation-8}) follows from {\bf LS} and {\bf LK2}.

      Therefore a consequence of successfully executing $\alpha$ is that James
      learns that Ed or Tim knows whether the grant application was successful.
  \end{example}

  \subsection{\classKFF{}}

  \begin{definition}[Axiomatisation \axiomAflKFF{}]\label{afl-kff-axioms}
      The axiomatisation \axiomAflKFF{} is a substitution schema consisting of the
      rules and axioms of \axiomKFF{} 
      along with the rules and axioms of \axiomAflK{},
      but substituting the \axiomAflK{} axiom {\bf LK1} for the axiom:
      $$
      \begin{array}{rl}
          {\bf LK1} & \proves \allacts{\learns_\agentsB (\alpha, \beta)} \necessary[\agentA] \chi \iff \necessary[\agentA] \allacts{\alpha \choice \beta} \chi \text{ for } \agentA \in \agentsB\\
      \end{array}
      $$
      and the rule:
      $$
      \begin{array}{rl}
          {\bf NecL} & \text{From $\proves \phi$ infer $\proves \allacts{\alpha} \phi$}
      \end{array}
      $$
      where $\chi$ is a $(\agents \setminus \{\agentA\})$-restricted formula.
  \end{definition}

  \begin{proposition}\label{afl-kff-axioms-soundness}
      The axiomatisation \axiomAflKFF{} is sound in the logic \logicAmlKFF{}.
  \end{proposition}

  \begin{proof}
      Soundness of {\bf LT}, {\bf LU}, {\bf LS}, {\bf LP}, {\bf LN}, {\bf LC},
      {\bf LK2} and {\bf NecL} follow from the same reasoning as in the proof
      of Proposition~\ref{afl-k-axioms-soundness}.

      {\bf LK1} follows from the \axiomAmlKFF{} axiom {\bf AK}.
      We note that as $\agentA \in \agentsB$, from Definition~\ref{afl-kff-learning}
      we have $\pointedActionModel{\actionStateS \actionAccessibilityAgent{\agentA}} \bisimilar_{(\agents \setminus \{\agentA\})} \tau(\alpha \choice \beta)$,
      and as $\chi$ is $(\agents \setminus \{\agentA\})$-restricted formula then
      $\entails \allacts{\pointedActionModel{\actionStateS \actionAccessibilityAgent{\agentA}}} \chi \iff \allacts{\tau(\alpha \choice \beta)} \chi$.
  \end{proof}

  \begin{proposition}\label{afl-kff-axioms-completeness}
      The axiomatisation \axiomAflKFF{} is complete for the logic \logicAmlKFF{}.
  \end{proposition}

  We note that the axiomatisation \axiomAflKFF{} forms a set of reduction
  axioms that gives a provably correct translation from \langAfl{} to \lang{}.
  To translate a subformula $\allacts{\alpha} \phi$, where $\phi \in \lang$, we
  must first translate $\phi$ to the alternating disjunctive normal form of
  \cite{hales2012}, which gives the property that for every subformula
  $\necessary[\agentA] \psi$, the formula $\psi$ is $(\agents \setminus
  \{\agentA\})$-restricted, and therefore {\bf LK1} is applicable.

  \subsection{\classS{}}

  \begin{definition}[Axiomatisation \axiomAflS{}]\label{afl-s-axioms}
      The axiomatisation \axiomAflS{} is a substitution schema consisting of the
      rules and axioms of \axiomS{} along with the axioms:
      $$
      \begin{array}{rl}
          {\bf LT} & \proves \allacts{\test{\phi}} \psi \iff (\phi \implies \psi) \text{ for } \psi \in \lang\\
          {\bf LU} & \proves \allacts{\alpha \choice \beta} \phi \iff (\allacts{\alpha} \phi \land \allacts{\beta} \phi)\\
          {\bf LS} & \proves \allacts{\alpha \compose \beta} \phi \iff \allacts{\alpha} \allacts{\beta} \phi\\
          {\bf LP} & \proves \allacts{\learns_\agentsB (\alpha, \beta)} \atomP \iff \atomP\\
          {\bf LN} & \proves \allacts{\learns_\agentsB (\alpha, \beta)} \neg \phi \iff \neg \allacts{\learns_\agentsB (\alpha, \beta)} \phi\\
          {\bf LC} & \proves \allacts{\learns_\agentsB (\alpha, \beta)} (\phi \land \psi) \iff (\allacts{\learns_\agentsB (\alpha, \beta)} \phi \land \allacts{\learns_\agentsB (\alpha, \beta)} \psi)\\
      \end{array}
      $$
      and the rule:
      $$
      \begin{array}{rl}
          {\bf NecL} & \text{From $\proves \phi$ infer $\proves \allacts{\alpha} \phi$}
      \end{array}
      $$
      where $\chi$ is a $(\agents \setminus \{\agentA\})$-restricted formula.
  \end{definition}

  \begin{proposition}\label{afl-s-axioms-soundness}
      The axiomatisation \axiomAflS{} is sound in the logic \logicAmlS{}.
  \end{proposition}

  \begin{proof}
      Soundness of {\bf LT}, {\bf LU}, {\bf LS}, {\bf LP}, {\bf LN}, {\bf LC}
      and {\bf NecL} follow from the same reasoning as in the proof of
      Proposition~\ref{afl-k-axioms-soundness}.
  \end{proof}

  We note that we do not have a axioms in \logicAflS{} corresponding to the
  axioms {\bf LK1} and {\bf LK2} from \axiomAflK{} and \axiomAflKFF{}.
  {\bf LK1} works in the setting of \classK{} because the $\agentsB$-successors
  of the root state of $\tau(\learns_\agentsB \alpha)$ are bisimilar to the
  root states of $\tau(\alpha)$, and so the consequences of executing
  $\tau(\alpha)$ are the same as the consequences of executing the
  $\agentsB$-successors of $\tau(\learns_\agentsB \alpha)$. In the setting of
  \classKFF{} this is not the case, however we do have the restricted property
  of $\agentsB$-bisimilarity, giving us that the $\agentsB$-restricted
  consequences are the same. In the setting of \classS{} we do not know of such
  a property to relate the consequences of the $\agentsB$-successors of
  $\tau(\learns_\agentsB (\alpha, \beta))$ to the consequences of $\tau(\alpha
  \choice \beta)$. Given the correspondence results of the previous section, it
  should be possible to construct an action formula that is $n$-bisimilar to
  the $\agentsB$-successors of $\tau(\learns_\agentsB (\alpha, \beta))$, where
  $d(\phi) = n$, and define axioms for {\bf LK1} and {\bf LK2} in terms of
  this action formula and not $\alpha \choice \beta$. However translating
  $\langAfl$ formulae into $\langAml$ formulae and then using the
  axiomatisation \axiomAmlS{} would certainly be simpler.

  \section{Correspondence}\label{correspondence}

  In the following subsections we show the correspondence between action
  formulae and action models in the settings of \classK{}, \classKFF{} and
  \classS{}. In each setting we show that action formulae are capable of
  representing any action model up to $n$-bisimilarity.

  \subsection{\classK{}}

  To begin we give two lemmas to simplify the construction that we will use
  for our correspondence result in \classK{}.

  \begin{lemma}\label{afl-k-construction-test}
      Let $\phi \in \langAfl$ and 
      $\pointedActionModel{\actionStateS} = \pointedActionModelTuple{\actionStateS} \in \classAM$.
      Then let $\pointedActionModel[\prime]{\actionStateS[\prime]} = \pointedActionModelTuple[\prime]{\actionStateS[\prime]} \in \classAM$ where:
      \begin{eqnarray*}
          \actionStates[\prime] &=& \actionStates \cup \{\actionStateS[\prime]\}\\
          \actionAccessibilityAgent[\prime]{\agentA} &=& \actionAccessibilityAgent{\agentA} \cup \{(\actionStateS[\prime], \actionStateT) \mid \actionStateT \in \actionStateS \actionAccessibilityAgent{\agentA}\} \text{ for } \agentA \in \agents\\
          \actionPrecondition[\prime] &=& \actionPrecondition \cup \{(\actionStateS[\prime], \phi \land \actionPrecondition(\actionStateS))\}
      \end{eqnarray*}
      Then $\tau(\test{\phi}) \exec \pointedActionModel{\actionStateS}  \bisimilar \pointedActionModel[\prime]{\actionStateS[\prime]}$.
  \end{lemma}

  \begin{lemma}\label{afl-k-construction-learning}
      Let $\alpha \in \langAflAct$ where $\tau(\alpha) = \pointedActionModel[\alpha]{\actionStatesT[\alpha]} = \pointedActionModelTuple[\alpha]{\actionStatesT[\alpha]}$,
      $\agentA \in \agents$ 
      and $\pointedActionModel{\actionStateS} = \pointedActionModelTuple{\actionStateS} \in \classAM$ 
      such that $\actionStateS \actionAccessibilityAgent{\agentA} = \{\actionStateT\}$ 
      for some $\actionStateT \in \actionStates$ 
      and $\actionStateT \actionAccessibilityAgent{\agentA} = \{\actionStateT\}$ 
      Then let $\pointedActionModel[\prime]{\actionStateS[\prime]} = \pointedActionModelTuple[\prime]{\actionStateS[\prime]} \in \classAM$ where:
      \begin{eqnarray*}
          \actionStates[\prime] &=& \actionStates \cup \actionStates[\alpha] \cup \{\actionStateS[\prime]\}\\
          \actionAccessibilityAgent[\prime]{\agentA} &=& \actionAccessibilityAgent{\agentA} \cup \actionAccessibilityAgent[\alpha]{\agentA} \cup \{(\actionStateS[\prime], \actionStateT[\alpha]) \mid \actionStateT[\alpha] \in \actionStatesT[\alpha]\}\\
          \actionAccessibilityAgent[\prime]{\agentB} &=& \actionAccessibilityAgent{\agentB} \cup \actionAccessibilityAgent[\alpha]{\agentB} \cup \{(\actionStateS[\prime], \actionStateT) \mid \actionStateT \in \actionStateS \actionAccessibilityAgent{\agentB}\} \text{ for } \agentB \in \agents \setminus \{\agentA\}\\
          \actionPrecondition[\prime] &=& \actionPrecondition \cup \{(\actionStateS[\prime], \actionPrecondition(\actionStateS))\}
      \end{eqnarray*}
      Then $\tau(\learns_\agentA \alpha) \exec \pointedActionModel{\actionStateS} \bisimilar \pointedActionModel[\prime]{\actionStateS[\prime]}$.
  \end{lemma}

  \begin{proposition}\label{afl-k-correspondence}
      Let $\pointedActionModel{\actionStateS} \in \classAM$ and 
      let $n \in \mathbb{N}$. 
      Then there exists $\alpha \in \langAflAct$ such that 
      $\pointedActionModel{\actionStateS} \bisimilar_n \tau(\alpha)$.
  \end{proposition}

  \begin{proof}
      By induction on $n$.

      Suppose that $n = 0$.
      Let $\alpha = \test{\actionPrecondition(\actionStateS)}$ and
      $\tau(\alpha) = \pointedActionModel[\prime]{\actionStateS[\prime]} = \pointedActionModelTuple[\prime]{\actionStateS[\prime]}$. 
      From Definition~\ref{afl-k-test} we have that
      $\actionPrecondition(\actionStateS) = \actionPrecondition[\prime](\actionStateS[\prime])$, so
      $(\pointedActionModel{\actionStateS}, \pointedActionModel[\prime]{\actionStateS[\prime]})$ satisfies 
      {\bf atoms} and therefore 
      $\pointedActionModel{\actionStateS} \bisimilar_0 \pointedActionModel[\prime]{\actionStateS[\prime]}$.

      Suppose that $n > 0$. 
      By the induction hypothesis, for every $\agentA \in \agents$, 
      $\actionStateT \in \actionStateS \actionAccessibilityAgent{\agentA}$ 
      there exists $\alpha^{\agentA,\actionStateT} \in \langAflAct$ such that 
      $\pointedActionModel{\actionStateT} \bisimilar_{(n - 1)} \tau(\alpha^{\agentA,\actionStateT})$,
      where $\tau(\alpha^{\agentA,\actionStateT}) \bisimilar \pointedActionModel[\agentA,\actionStateT]{\actionStateS[\agentA,\actionStateT]} = \pointedActionModelTuple[\agentA,\actionStateT]{\actionStateS[\agentA,\actionStateT]}$.
      
      Let $\alpha = \test{\actionPrecondition(\actionStateS)} \compose \bigcompose_{\agentA \in \agents} \learns_\agentA (\bigchoice_{\actionStateT \in \actionStateS \actionAccessibilityAgent{\agentA}} \alpha^{\actionStateT})$. 
      Then from Lemmas~\ref{afl-k-construction-test} and~\ref{afl-k-construction-learning}: $\tau(\alpha) \bisimilar \pointedActionModel[\prime]{\actionStateS[\prime]} = \pointedActionModelTuple[\prime]{\actionStateS[\prime]}$ where:
      \begin{eqnarray*}
          \actionStates[\prime] &=& \bigcup_{\agentA \in \agents, \actionStateT \in \actionStateS \actionAccessibilityAgent{\agentA}} (\actionStates[\agentA,\actionStateT]) \cup \{\actionStateS[\prime]\}\\
          \actionAccessibilityAgent[\prime]{\agentA} &=& \bigcup_{\agentB \in \agents, \actionStateT \in \actionStateS \actionAccessibilityAgent{\agentB}} (\actionAccessibilityAgent[\agentB,\actionStateT]{\agentA}) \cup \{(\actionStateS[\prime], \actionStateS[\agentA,\actionStateT]) \mid \actionStateT \in \actionStateS \actionAccessibilityAgent{\agentA}\} \text{ for } \agentA \in \agents\\
          \actionPrecondition[\prime] &=& \bigcup_{\agentA \in \agents, \actionStateT \in \actionStateS \actionAccessibilityAgent{\agentA}} (\actionPrecondition[\agentA,\actionStateT]) \cup \{(\actionStateS[\prime], \actionPrecondition(\actionStateS))\}
      \end{eqnarray*}

      We note for every $\agentA \in \agents$, 
      $\actionStateT \in \actionStateS \actionAccessibilityAgent{\agentA}$ that
      $\pointedActionModel[\prime]{\actionStateS[\agentA,\actionStateT]} \bisimilar \pointedActionModel[\agentA,\actionStateT]{\actionStateS[\agentA,\actionStateT]}$
      as for every $\agentA \in \agents$, 
      $\actionStateU \in \actionStates[\agentA,\actionStateT]$ we have 
      $\actionStateU \actionAccessibilityAgent[\prime]{\agentA} = \actionStateU \actionAccessibilityAgent[\agentA,\actionStateT]{\agentA}$.
      
      We show that $(\pointedActionModel{\actionStateS}, \pointedActionModel[\prime]{\actionStateS[\prime]})$
      satisfies {\bf atoms}, {\bf forth-$n$-$\agentA$} and {\bf back-$n$-$\agentA$} 
      for every $\agentA \in \agents$.

      \paragraph{atoms} By construction
      $\actionPrecondition[\prime](\actionStateS[\prime]) =
      \actionPrecondition(\actionStateS)$.

      \paragraph{forth-$n$-$\agentA$} 
      Let $\actionStateT \in \actionStateS \actionAccessibilityAgent{\agentA}$.
      By construction $\actionStateS[\agentA,\actionStateT] \in \actionStateS[\prime] \actionAccessibilityAgent[\prime]{\agentA}$, 
      by the induction hypothesis $\pointedActionModel{\actionStateT} \bisimilar_{(n - 1)} \pointedActionModel[\agentA,\actionStateT]{\actionStateS[\agentA,\actionStateT]}$ 
      and from above $\pointedActionModel[\agentA,\actionStateT]{\actionStateS[\agentA,\actionStateT]} \bisimilar \pointedActionModel[\prime]{\actionStateS[\agentA,\actionStateT]}$.
      Therefore by transitivity $\pointedActionModel{\actionStateT} \bisimilar_{(n - 1)} \pointedActionModel[\prime]{\actionStateS[\agentA,\actionStateT]}$.

      \paragraph{back-$n$-$\agentA$} Follows from similar reasoning to {\bf forth-$n$-$\agentA$}.

      Therefore $\pointedActionModel{\actionStateS} \bisimilar_n \tau(\alpha)$.
  \end{proof}

  \begin{corollary}\label{afl-k-correspondence-aml-allacts}
      Let $\pointedActionModel{\actionStateS} \in \classAM$.
      Then for every $\phi \in \langAml$
      there exists $\alpha \in \langAflAct$ 
      such that $\entails_\logicAmlK{} \allacts{\pointedActionModel{\actionStateS}} \phi \iff \allacts{\tau(\alpha)} \phi$.
  \end{corollary}

  \begin{proof}
      Suppose that $d(\phi) = n$.
      From Proposition~\ref{afl-k-correspondence} 
      there exists $\alpha \in \langAflAct$ 
      such that $\pointedActionModel{\actionStateS} \bisimilar_n \tau(\alpha)$.
      Therefore for every $\pointedModel{\actionStateS} \in \classK$
      we have $\pointedModel{\actionStateS} \exec \pointedActionModel{\actionStateS} \bisimilar_n \pointedModel{\actionStateS} \exec \tau(\alpha)$
      and so $\pointedModel{\actionStateS} \exec \pointedActionModel{\actionStateS} \entails_\logicAmlK{} \phi$
      if and only if $\pointedModel{\actionStateS} \exec \tau(\alpha) \entails_\logicAmlK{} \phi$.
      Therefore $\pointedModel{\actionStateS} \entails_\logicAmlK{} \allacts{\pointedActionModel{\actionStateS}}  \phi$
      if and only if $\pointedModel{\actionStateS} \entails_\logicAmlK{} \allacts{\tau(\alpha)}  \phi$.
  \end{proof}

  \begin{corollary}\label{afl-k-correspondence-afl-aml}
      Let $\phi \in \langAml$. 
      Then there exists $\phi' \in \langAfl$ 
      such that for every $\pointedModel{\stateS} \in \classK$: 
      $\pointedModel{\stateS} \entails_\logicAmlK{} \phi$ if and only if
      $\pointedModel{\stateS} \entails_\logicAflK{} \phi'$.
  \end{corollary}

  \begin{proof}[Sketch]
      Given Corollary~\ref{afl-k-correspondence-aml-allacts} we can replace all
      occurrences of $\allacts{\pointedActionModel{\actionStateS}} \psi$ 
      within $\phi$ with an equivalent $\allacts{\alpha} \psi$ 
      where $\alpha \in \langAflAct$.
  \end{proof}

  \subsection{\classKFF{}}

  As in the previous subsection we give a lemma to simplify the construction
  that we will use, although as the definition of $\tau(\test{\phi})$ is the
  same between \classK{} and \classKFF{} we simply reuse
  Lemma~\ref{afl-k-construction-test} from the previous subsection.

  \begin{lemma}\label{afl-kff-construction-learning}
      Let $\agentA \in \agents$,
      $\alpha \in \langAflAct$ where $\tau(\alpha) = \pointedActionModel[\alpha]{\actionStatesT[\alpha]} = \pointedActionModelTuple[\alpha]{\actionStatesT[\alpha]}$,
      and $\pointedActionModel{\actionStateS} = \pointedActionModelTuple{\actionStateS} \in \classAM$ 
      such that $\actionStateS \actionAccessibilityAgent{\agentA} = \{\actionStateT\}$ 
      for some $\actionStateT \in \actionStates$ 
      and $\actionStateT \actionAccessibilityAgent{\agentA} = \{\actionStateT\}$ 
      Then let $\pointedActionModel[\prime]{\actionStateS[\prime]} = \pointedActionModelTuple[\prime]{\actionStateS[\prime]} \in \classAM$ where:
      \begin{eqnarray*}
          \actionStates[\prime] &=& \actionStates \cup \actionStates[\alpha] \cup \{\proxyStateT[\alpha] \mid \actionStateT[\alpha] \in \actionStatesT[\alpha]\} \cup \{\actionStateS[\prime]\}\\
          \actionAccessibilityAgent[\prime]{\agentA} &=& \actionAccessibilityAgent{\agentA} \cup \actionAccessibilityAgent[\alpha]{\agentA} \cup \{(\actionStateS[\prime], \proxyStateT[\alpha]) \mid \actionStateT[\alpha] \in \actionStatesT[\alpha]\} \cup \{(\proxyStateT[\alpha], \proxyStateU[\alpha]) \mid \actionStateT[\alpha], \actionStateU[\alpha] \in \actionStatesT[\alpha]\}\\
          \actionAccessibilityAgent[\prime]{\agentB} &=& \actionAccessibilityAgent{\agentB} \cup \actionAccessibilityAgent[\alpha]{\agentB} \cup \{(\actionStateS[\prime], \actionStateT) \mid \actionStateT \in \actionStateS \actionAccessibilityAgent{\agentB}\} \text{ for } \agentB \in \agents \setminus \{\agentA\}\\
          \actionPrecondition[\prime] &=& \actionPrecondition \cup \{(\proxyStateT[\alpha], \actionPrecondition[\alpha](\actionStateT[\alpha])) \mid \actionStateT[\alpha] \in \actionStatesT[\alpha]\} \cup \{(\actionStateS[\prime], \actionPrecondition(\actionStateS))\}
      \end{eqnarray*}
      Then $\tau(\learns_\agentA \alpha) \exec \pointedActionModel{\actionStateS} \bisimilar \pointedActionModel[\prime]{\actionStateS[\prime]}$.
  \end{lemma}

  \begin{proposition}\label{afl-kff-correspondence}
      Let $\pointedActionModel{\actionStateS} \in \classAM_\classKFF$ and 
      let $n \in \mathbb{N}$. 
      Then there exists $\alpha \in \langAflAct$ such that 
      $\pointedActionModel{\actionStateS} \bisimilar_n \tau(\alpha)$.
  \end{proposition}

  \begin{proof}
      By induction on $n$.

      Suppose that $n = 0$. 
      Let $\alpha = \test{\actionPrecondition(\actionStateS)}$ and
      $\tau(\alpha) = \pointedActionModel[\prime]{\actionStateS[\prime]} = \pointedActionModelTuple[\prime]{\actionStateS[\prime]}$. 
      From Definition~\ref{afl-kff-test} we have that
      $\actionPrecondition(\actionStateS) = \actionPrecondition[\prime](\actionStateS[\prime])$, so
      $(\pointedActionModel{\actionStateS}, \pointedActionModel[\prime]{\actionStateS[\prime]})$ satisfies 
      {\bf atoms} and therefore 
      $\pointedActionModel{\actionStateS} \bisimilar_0 \pointedActionModel[\prime]{\actionStateS[\prime]}$.

      Suppose that $n > 0$. 
      By the induction hypothesis, for every $\agentA \in \agents$, 
      $\actionStateT \in \actionStateS \actionAccessibilityAgent{\agentA}$ 
      there exists $\alpha^{\agentA,\actionStateT} \in \langAflAct$ such that 
      $\pointedActionModel{\actionStateT} \bisimilar_{(n - 1)} \tau(\alpha^{\agentA,\actionStateT})$. 
      For every $\agentA \in \agents$, 
      $\actionStateT \in \actionStateS \actionAccessibilityAgent{\agentA}$ 
      let $\tau(\alpha^{\agentA,\actionStateT}) = \pointedActionModel[\agentA,\actionStateT]{\actionStateS[\agentA,\actionStateT]} = \pointedActionModelTuple[\agentA,\actionStateT]{\actionStateS[\agentA,\actionStateT]}$.
      
      Let $\alpha = \test{\actionPrecondition(\actionStateS)} \compose \bigcompose_{\agentA \in \agents} \learns_\agentA (\bigchoice_{\actionStateT \in \actionStateS \actionAccessibilityAgent{\agentA}} \alpha^{\agentA,\actionStateT})$. 
      Then from Lemmas~\ref{afl-k-construction-test} and~\ref{afl-kff-construction-learning}: $\tau(\alpha) \bisimilar \pointedActionModel[\prime]{\actionStateS[\prime]} = \pointedActionModelTuple[\prime]{\actionStateS[\prime]}$ where:
      \begin{eqnarray*}
          \actionStates[\prime] &=& \bigcup_{\agentA \in \agents, \actionStateT \in \actionStateS \actionAccessibilityAgent{\agentA}} (\actionStates[\agentA,\actionStateT]) \cup \{\proxyStateS[\agentA,\actionStateT] \mid \agentA \in \agents, \actionStateT \in \actionStateS \actionAccessibilityAgent{\agentA}\} \cup \{\actionStateS[\prime]\}\\
          \actionAccessibilityAgent[\prime]{\agentA} &=& \bigcup_{\agentB \in \agents, \actionStateT \in \actionStateS \actionAccessibilityAgent{\agentB}} (\actionAccessibilityAgent[\agentB,\actionStateT]{\agentA}) \cup \{(\actionStateS[\prime], \proxyStateS[\agentA, \actionStateT]) \mid \actionStateT \in \actionStateS \actionAccessibilityAgent{\agentA}\} \cup \{(\proxyStateS[\agentA, \actionStateT], \proxyStateS[\agentA, \actionStateU]) \mid \actionStateT, \actionStateU \in \actionStateS \actionAccessibilityAgent{\agentA}\} \cup\\
                                                     &&\hspace{45pt}\{(\proxyStateS[\agentB, \actionStateT], \actionStateU) \mid \agentB \in \agents \setminus \{\agentA\}, \actionStateT \in \actionStateS \actionAccessibilityAgent{\agentB}, \actionStateU \in \actionStateS[\agentB,\actionStateT] \actionAccessibilityAgent[\agentB,\actionStateT]{\agentA}\} \text{ for } \agentA \in \agents\\
          \actionPrecondition[\prime] &=& \bigcup_{\agentA \in \agents, \actionStateT \in \actionStateS \actionAccessibilityAgent{\agentA}} (\actionPrecondition[\agentA,\actionStateT]) \cup \{(\proxyStateS[\agentA, \actionStateT], \actionPrecondition[\agentA,\actionStateT](\actionStateS[\agentA,\actionStateT])) \mid \agentA \in \agents, \actionStateT \in \actionStateS \actionAccessibilityAgent{\agentA}\} \cup \{(\actionStateS[\prime], \actionPrecondition(\actionStateS))\}
      \end{eqnarray*}

      As in the proof of Proposition~\ref{afl-k-correspondence}, 
      we note for every $\agentA \in \agents$, 
      $\actionStateT \in \actionStateS \actionAccessibilityAgent{\agentA}$ that
      $\pointedActionModel[\prime]{\actionStateS[\agentA,\actionStateT]} \bisimilar \pointedActionModel[\agentA,\actionStateT]{\actionStateS[\agentA,\actionStateT]}$.

      We need to show that $(\pointedActionModel{\actionStateS}, \pointedActionModel[\prime]{\actionStateS[\prime]})$ 
      satisfies {\bf atoms}, {\bf forth-$n$-$\agentA$} and {\bf back-$n$-$\agentA$} 
      for every $\agentA \in \agents$.
      We use reasoning similar to the proof of Proposition~\ref{afl-k-correspondence}, 
      however noting that the successors of $\actionStateS[\prime]$ in
      $\actionModel[\prime]$ are not the same as in the construction used
      previously.
      We claim that each $\proxyStateS[\agentA,\actionStateT]$
      state is $(n-1)$-bisimilar to the corresponding $\actionStateS[\agentA,\actionStateT]$ state.
      We show this by showing for every $0 \leq i \leq n - 1$, 
      $\agentA \in \agents$,
      $\actionStateT \in \actionStateS \actionAccessibilityAgent{\agentA}$ that 
      $\pointedActionModel[\prime]{\proxyStateS[\agentA,\actionStateT]} \bisimilar_i \pointedActionModel[\prime]{\actionStateS[\agentA,\actionStateT]}$.
      We proceed by induction on $i$.

      \paragraph{atoms} By construction $\actionPrecondition[\prime](\proxyStateS[\agentA, \actionStateT]) = \actionPrecondition[\prime](\actionStateS[\agentA,\actionStateT])$.

      \paragraph{forth-$i$-$\agentB$} Suppose that $0 < i \leq n - 1$. Let $\actionStateU \in \proxyStateS[\agentA,\actionStateT] \actionAccessibilityAgent[\prime]{\agentB}$. 
      
      Suppose that $\agentB = \agentA$. 
      By construction there exists $\actionStateV \in \actionStateS \actionAccessibilityAgent{\agentA}$ 
      such that $\actionStateU = \proxyStateS[\agentA,\actionStateV]$.
      From above $\pointedActionModel[\prime]{\actionStateS[\agentA,\actionStateT]} \bisimilar \pointedActionModel[\agentA,\actionStateT]{\actionStateS[\agentA,\actionStateT]}$
      and $\pointedActionModel[\prime]{\actionStateS[\agentA,\actionStateV]} \bisimilar \pointedActionModel[\agentA,\actionStateV]{\actionStateS[\agentA,\actionStateV]}$.
      By the outer induction hypothesis $\pointedActionModel[\agentA,\actionStateT]{\actionStateS[\agentA,\actionStateT]} \bisimilar_{(n - 1)} \pointedActionModel{\actionStateT}$
      and $\pointedActionModel[\agentA,\actionStateV]{\actionStateS[\agentA,\actionStateV]} \bisimilar_{(n - 1)} \pointedActionModel{\actionStateV}$.
      By transitivity $\pointedActionModel[\prime]{\actionStateS[\agentA,\actionStateT]} \bisimilar_{(n - 1)} \pointedActionModel{\actionStateT}$
      and $\pointedActionModel[\prime]{\actionStateS[\agentA,\actionStateV]} \bisimilar_{(n - 1)} \pointedActionModel{\actionStateV}$.
      As $\actionStateV \in \actionStateT \actionAccessibilityAgent{\agentA}$ from {\bf back-$(n - 1)$-$\agentA$}
      there exists $\actionStateW \in \actionStateS[\agentA,\actionStateT] \actionAccessibilityAgent[\prime]{\agentA}$
      such that $\pointedActionModel[\prime]{\actionStateW} \bisimilar_{(n - 2)} \pointedActionModel{\actionStateV}$.
      By transitivity $\pointedActionModel[\prime]{\actionStateW} \bisimilar_{(n - 2)} \pointedActionModel[\prime]{\actionStateS[\agentA,\actionStateV]}$.
      By the induction hypothesis $\pointedActionModel[\prime]{\proxyStateS[\agentA,\actionStateV]} \bisimilar_{(i - 1)} \pointedActionModel[\prime]{\actionStateS[\agentA,\actionStateV]}$.
      Therefore by transitivity $\pointedActionModel[\prime]{\proxyStateS[\agentA,\actionStateV]} \bisimilar_{(i - 1)} \pointedActionModel[\prime]{\actionStateW}$.

      Suppose that $\agentB \neq \agentA$. By construction
      $\proxyStateS[\agentA,\actionStateT] \actionAccessibilityAgent[\prime]{\agentB} = \actionStateS[\agentA,\actionStateT] \actionAccessibilityAgent[\prime]{\agentB}$,
      so $\actionStateU \in \actionStateS[\agentA,\actionStateT] \actionAccessibilityAgent[\prime]{\agentB}$ 
      and we trivially have that $\pointedActionModel[\prime]{\actionStateU} \bisimilar \pointedActionModel[\prime]{\actionStateU}$.

      \paragraph{back-$i$-$\agentB$} Follows similar reasoning to {\bf forth-$i$-$\agentB$}.

      Therefore for every $\agentA \in \agents$, 
      $\actionStateT \in \actionStateS \actionAccessibilityAgent{\agentA}$ we have that
      $\pointedActionModel[\prime]{\proxyStateS[\agentA,\actionStateT]} \bisimilar_{(n - 1)} \pointedActionModel[\prime]{\actionStateS[\agentA,\actionStateT]}$.

      We can now show that $\pointedActionModel{\actionStateS} \bisimilar_n \pointedActionModel[\prime]{\actionStateS[\prime]}$ 
      by using the same reasoning as the proof for Proposition~\ref{afl-k-correspondence}, using the $(n - 1)$-bisimilar
      $\pointedActionModel[\prime]{\proxyStateS[\agentA,\actionStateT]}$ states 
      in place of corresponding $\pointedActionModel[\prime]{\actionStateS[\agentA,\actionStateT]}$ states.

      Therefore $\pointedActionModel{\actionStateS} \bisimilar_n \tau(\alpha)$.
  \end{proof}

  \begin{corollary}
      Let $\pointedActionModel{\actionStateS} \in \classAM_\classKFF$.
      Then for every $\phi \in \langAml$
      there exists $\alpha \in \langAflAct$ 
      such that $\entails_\logicAmlKFF{} \allacts{\pointedActionModel{\actionStateS}} \phi \iff \allacts{\tau(\alpha)} \phi$.
  \end{corollary}

  \begin{corollary}
      Let $\phi \in \langAml$. 
      Then there exists $\phi' \in \langAfl$ 
      such that for every $\pointedModel{\stateS} \in \classKFF$: 
      $\pointedModel{\stateS} \entails_\logicAmlKFF{} \phi$ if and only if
      $\pointedModel{\stateS} \entails_\logicAflKFF{} \phi'$.
  \end{corollary}

  \subsection{\classS{}}

  Once more we give two lemmas to simplify the construction that we will use.

  \begin{lemma}\label{afl-s-construction-test}
      Let $\phi \in \langAfl$ and 
      $\pointedActionModel{\actionStateS} = \pointedActionModelTuple{\actionStateS} \in \classAM$.
      Then let $\pointedActionModel[\prime]{\actionStateS} = \pointedActionModelTuple[\prime]{\actionStateS} \in \classAM$ where:
      \begin{eqnarray*}
          \actionStates[\prime] &=& \actionStates\\
          \actionAccessibilityAgent[\prime]{\agentA} &=& \actionAccessibilityAgent{\agentA} \text{ for } \agentA \in \agents\\
          \actionPrecondition[\prime] &=& \actionPrecondition \setminus \{(\actionStateS, \actionPrecondition(\actionStateS))\} \cup \{(\actionStateS, \phi \land \actionPrecondition(\actionStateS))\}
      \end{eqnarray*}
      Then $\tau(\test{\phi}) \exec \pointedActionModel{\actionStateS} \bisimilar \pointedActionModel[\prime]{\actionStateS[\prime]}$.
  \end{lemma}

  \begin{lemma}\label{afl-s-construction-learning}
      Let $\agentA \in \agents$,
      $\alpha \in \langAflAct$ where $\tau(\alpha) = \pointedActionModel[\alpha]{\actionStatesT[\alpha]} = \pointedActionModelTuple[\alpha]{\actionStatesT[\alpha]}$,
      and $\pointedActionModel{\actionStateS} = \pointedActionModelTuple{\actionStateS} \in \classAM$ 
      such that $\actionStateS \actionAccessibilityAgent{\agentA} = \{\actionStateS\}$ 
      and $\actionPrecondition(\actionStateS) = \top$.
      Then let $\pointedActionModel[\prime]{\actionStateS} = \pointedActionModelTuple[\prime]{\actionStateS} \in \classAM$ where:
      \begin{eqnarray*}
          \actionStates[\prime] &=& \actionStates \cup \actionStates[\alpha] \cup \{\proxyStateT[\alpha] \mid \actionStateT[\alpha] \in \actionStatesT[\alpha]\}\\
          \actionAccessibilityAgent[\prime]{\agentA} &=& \actionAccessibilityAgent{\agentA} \cup \actionAccessibilityAgent[\alpha]{\agentA} \cup (\{\actionStateS\} \cup \{\proxyStateT[\alpha] \mid \actionStateT[\alpha] \in \actionStatesT[\alpha]\})^2\\
          \actionAccessibilityAgent[\prime]{\agentB} &=& \actionAccessibilityAgent{\agentB} \cup \actionAccessibilityAgent[\alpha]{\agentB} \cup (\{\proxyStateT[\alpha]\} \cup \actionStateT[\alpha] \actionAccessibilityAgent[\alpha]{\agentB})^2 \text{ for } \agentB \in \agents \setminus \{\agentA\}\\
          \actionPrecondition[\prime] &=& \actionPrecondition \cup \{(\proxyStateT[\alpha], \actionPrecondition[\alpha](\actionStateT[\alpha])) \mid \actionStateT[\alpha] \in \actionStatesT[\alpha]\}
      \end{eqnarray*}
      Then $\tau(\learns_\agentA (\test{\top}, \alpha)) \exec \pointedActionModel{\actionStateS} \bisimilar \pointedActionModel[\prime]{\actionStateS[\prime]}$.
  \end{lemma}

  \begin{proposition}\label{afl-s-correspondence}
      Let $\pointedActionModel{\actionStateS} \in \classAM_\classS$ and 
      let $n \in \mathbb{N}$. 
      Then there exists $\alpha \in \langAflAct$ such that 
      $\pointedActionModel{\actionStateS} \bisimilar_n \tau(\alpha)$.
  \end{proposition}

  \begin{proof}
      By induction on $n$.

      Suppose that $n = 0$. 
      Let $\alpha = \test{\actionPrecondition(\actionStateS)}$ and
      $\tau(\alpha) = \pointedActionModel[\prime]{\actionStateS[\prime]} = \pointedActionModelTuple[\prime]{\actionStateS[\prime]}$. 
      From Definition~\ref{afl-s-test} we have that
      $\actionPrecondition(\actionStateS) = \actionPrecondition[\prime](\actionStateS[\prime])$, so
      $(\pointedActionModel{\actionStateS}, \pointedActionModel[\prime]{\actionStateS[\prime]})$ satisfies 
      {\bf atoms} and therefore 
      $\pointedActionModel{\actionStateS} \bisimilar_0 \pointedActionModel[\prime]{\actionStateS[\prime]}$.

      Suppose that $n > 0$. 
      By the induction hypothesis, for every $\agentA \in \agents$, 
      $\actionStateT \in \actionStateS \actionAccessibilityAgent{\agentA}$ 
      there exists $\alpha^{\agentA,\actionStateT} \in \langAflAct$ such that 
      $\pointedActionModel{\actionStateT} \bisimilar_{(n - 1)} \tau(\alpha^{\agentA,\actionStateT})$. 
      For every $\agentA \in \agents$, 
      $\actionStateT \in \actionStateS \actionAccessibilityAgent{\agentA}$ 
      let $\tau(\alpha^{\agentA,\actionStateT}) = \pointedActionModel[\agentA,\actionStateT]{\actionStateS[\agentA,\actionStateT]} = \pointedActionModelTuple[\agentA,\actionStateT]{\actionStateS[\agentA,\actionStateT]}$.
      
      Let $\alpha = \test{\actionPrecondition(\actionStateS)} \compose \bigcompose_{\agentA \in \agents} \learns_\agentA (\test{\top}, \bigchoice_{\actionStateT \in \actionStateS \actionAccessibilityAgent{\agentA}} \alpha^{\agentA,\actionStateT})$.
      Then from Lemmas~\ref{afl-k-construction-test} and~\ref{afl-kff-construction-learning}: $\tau(\alpha) = \pointedActionModel[\prime]{\actionStateS[\prime]} =\pointedActionModelTuple[\prime]{\actionStateS[\prime]}$ where:
      \begin{eqnarray*}
          \actionStates[\prime] &=& \bigcup_{\agentA \in \agents, \actionStateT \in \actionStateS \actionAccessibilityAgent{\agentA}} (\actionStates[\agentA,\actionStateT]) \cup \{\proxyStateS[\agentA,\actionStateT] \mid \agentA \in \agents, \actionStateT \in \actionStateS \actionAccessibilityAgent{\agentA}\} \cup \{\actionStateS[\prime]\}\\
          \actionAccessibilityAgent[\prime]{\agentA} &=& \bigcup_{\agentB \in \agents, \actionStateT \in \actionStateS \actionAccessibilityAgent{\agentB}} (\actionAccessibilityAgent[\agentB,\actionStateT]{\agentA}) \cup (\{\actionStateS[\prime]\} \cup \{\proxyStateS[\agentA,\actionStateT] \mid \actionStateT \in \actionStateS \actionAccessibilityAgent{\agentA}\})^2 \cup\\&&\quad\bigcup_{\agentB \in \agents \setminus \{\agentA\}, \actionStateT \in \actionAccessibilityAgent{\agentB}} (\{\proxyStateS[\agentB,\actionStateT]\} \cup \actionStateS[\agentB,\actionStateT] \actionAccessibilityAgent[\agentB,\actionStateT]{\agentA})^2 \text{ for } \agentA \in \agents\\
          \actionPrecondition[\prime] &=& \bigcup_{\agentA \in \agents, \actionStateT \in \actionStateS \actionAccessibilityAgent{\agentA}} (\actionPrecondition[\agentA,\actionStateT]) \cup \{(\proxyStateS[\agentA,\actionStateT], \actionPrecondition[\agentA,\actionStateT](\actionStateS[\agentA,\actionStateT])) \mid \agentA \in \agents, \actionStateT \in \actionStateS \actionAccessibilityAgent{\agentA}\} \cup \{(\actionStateS[\prime], \actionPrecondition(\actionStateS))\}
      \end{eqnarray*}

      We note that unlike the constructions used for
      Proposition~\ref{afl-k-correspondence} and
      Proposition~\ref{afl-kff-correspondence}, this construction does not have
      $\pointedActionModel[\prime]{\actionStateU} \bisimilar \pointedActionModel[\agentA,\actionStateT]{\actionStateU}$,
      as we do not have that 
      $\actionStateS[\agentA,\actionStateT] \actionAccessibilityAgent[\prime]{\agentA} = \actionStateS[\agentA,\actionStateT] \actionAccessibilityAgent[\agentA,\actionStateT]{\agentA}$.
      Similar to the proof of Proposition~\ref{afl-kff-correspondence} we claim
      that each $\proxyStateS[\agentA,\actionStateT]$ state is
      $(n-1)$-bisimilar to the corresponding $\actionStateS[\actionStateT]$ state. 
      However in lieu of bisimilarity of $\actionStates[\agentA,\actionStateT]$ 
      states we need another result for these states.
      We also need to consider the additional state $\actionStateS[\prime]$,
      which due to reflexivity is also a successor of itself.

      We need to show for every $0 \leq i \leq n - 1$: 
      \begin{enumerate}
          \item For every $\agentA \in \agents$: $\pointedActionModel[\prime]{\actionStateS[\prime]} \bisimilar_i \pointedActionModel[\prime]{\proxyStateS[\agentA,\actionStateS]}$.
          \item For every $\agentA \in \agents$, $\actionStateT \in \actionStateS \actionAccessibilityAgent{\agentA}$: $\pointedActionModel[\prime]{\proxyStateS[\agentA,\actionStateT]} \bisimilar_i \pointedActionModel[\prime]{\actionStateS[\agentA,\actionStateT]}$.
          \item For every $\agentA \in \agents$, $\actionStateT \in \actionStateS \actionAccessibilityAgent{\agentA}$, $\actionStateU \in \actionStates[\agentA,\actionStateT]$, $\actionStateV \in \actionStates$: if $\pointedActionModel[\agentA,\actionStateT]{\actionStateU} \bisimilar_i \pointedActionModel{\actionStateV}$ then $\pointedActionModel[\prime]{\actionStateU} \bisimilar_i \pointedActionModel{\actionStateV}$.
      \end{enumerate}

      We proceed by induction on $i$.

      \begin{enumerate}
          \item 
              For every $\agentA \in \agents$: $\pointedActionModel[\prime]{\actionStateS[\prime]} \bisimilar_i \pointedActionModel[\prime]{\proxyStateS[\agentA,\actionStateS]}$.

              \paragraph{atoms} By the outer induction hypothesis $\pointedActionModel[\agentA,\actionStateS]{\actionStateS[\agentA,\actionStateS]} \bisimilar_{(n - 1)} \pointedActionModel{\actionStateS}$
              and so $\proves \actionPrecondition[\agentA,\actionStateS](\actionStateS[\agentA,\actionStateS]) \iff \actionPrecondition(\actionStateS)$.
              By construction $\actionPrecondition[\prime](\actionStateS[\prime]) = \actionPrecondition(\actionStateS)$ 
              and $\actionPrecondition[\prime](\proxyStateS[\agentA,\actionStateS]) = \actionPrecondition[\agentA,\actionStateS](\actionStateS[\agentA,\actionStateS])$ 
              and therefore $\proves \actionPrecondition[\prime](\actionStateS[\prime]) \iff \actionPrecondition[\agentA,\actionStateS](\proxyStateS[\agentA,\actionStateS])$.

              \paragraph{forth-$i$-$\agentB$} Suppose that $0 < i \leq n - 1$. Let $\actionStateU \in \actionStateS[\prime] \actionAccessibilityAgent[\prime]{\agentA}$. 

              Suppose that $\agentB = \agentA$.
              By construction $\actionStateS[\prime] \actionAccessibilityAgent[\prime]{\agentA} = \proxyStateS[\agentA,\actionStateS] \actionAccessibilityAgent[\prime]{\agentA}$ 
              and we trivially have that $\pointedActionModel[\prime]{\actionStateU} \bisimilar \pointedActionModel[\prime]{\actionStateU}$.

              Suppose that $\agentB \neq \agentA$.
              By construction $\actionStateS[\prime] \actionAccessibilityAgent[\prime]{\agentB} = \{\proxyStateS[\agentB,\actionStateT] \mid \actionStateT \in \actionStateS \actionAccessibilityAgent{\agentB}\} \cup \{\actionStateS[\prime]\}$ and $\proxyStateS[\agentA,\actionStateS] \actionAccessibilityAgent[\prime]{\agentB} = \actionStateS[\agentA,\actionStateS] \actionAccessibilityAgent[\agentA,\actionStateS]{\agentB} \cup \{\proxyStateS[\agentA,\actionStateS]\}$. 
              Suppose that $\actionStateU = \actionStateS[\prime]$. 
              Then by the induction hypothesis $\pointedActionModel[\prime]{\actionStateS[\prime]} \bisimilar_{(i-1)} \pointedActionModel[\prime]{\proxyStateS[\agentA,\actionStateS]}$.
              Suppose that $\actionStateU \in \{\proxyStateS[\agentB,\actionStateT] \mid \actionStateT \in \actionStateS \actionAccessibilityAgent{\agentB}\}$. 
              Then there exists $\actionStateT \in \actionStateS \actionAccessibilityAgent{\agentB}$ 
              such that $\actionStateU = \proxyStateS[\agentB,\actionStateT]$.
              By the outer induction hypothesis $\pointedActionModel[\agentA,\actionStateS]{\actionStateS[\agentA,\actionStateS]} \bisimilar_{(n - 1)} \pointedActionModel{\actionStateS}$.
              As $\actionStateT \in \actionStateS \actionAccessibilityAgent{\agentB}$ then by {\bf back-$(n-1)$-$\agentB$}
              there exists $\actionStateV \in \actionStateS[\agentA,\actionStateS] \actionAccessibilityAgent[\agentA,\actionStateS]{\agentB} \subseteq \proxyStateS[\agentA,\actionStateS] \actionAccessibilityAgent[\prime]{\agentB}$
              such that $\pointedActionModel[\agentA,\actionStateS]{\actionStateV} \bisimilar_{(n - 2)} \pointedActionModel{\actionStateT}$.
              Then by the inner induction hypothesis this implies $\pointedActionModel[\prime]{\actionStateV} \bisimilar_{(i - 1)} \pointedActionModel{\actionStateT}$.
              By the inner induction hypothesis $\pointedActionModel[\prime]{\proxyStateS[\agentB,\actionStateT]} \bisimilar_{(i - 1)} \pointedActionModel[\prime]{\actionStateS[\agentB,\actionStateT]} \bisimilar_{(i - 1)} \pointedActionModel[\agentB,\actionStateT]{\actionStateS[\agentB,\actionStateT]}$ 
              and by the outer induction hypothesis $\pointedActionModel[\agentB,\actionStateT]{\actionStateS[\agentB,\actionStateT]} \bisimilar_{(n - 1)} \pointedActionModel{\actionStateT}$ 
              so by transitivity $\pointedActionModel[\prime]{\proxyStateS[\agentB,\actionStateT]} \bisimilar_{(i - 1)} \pointedActionModel{\actionStateT}$.
              Therefore by transitivity we have that $\pointedActionModel[\prime]{\proxyStateS[\agentB,\actionStateT]} \bisimilar_{(i - 1)} \pointedActionModel[\prime]{\actionStateV}$.

              \paragraph{back-$i$-$\agentB$} Follows similar reasoning to {\bf forth-$i$-$\agentB$}.

          \item 
              For every $\agentA \in \agents$, $\actionStateT \in \actionStateS \actionAccessibilityAgent{\agentA}$: $\pointedActionModel[\prime]{\proxyStateS[\agentA,\actionStateT]} \bisimilar_i \pointedActionModel[\prime]{\actionStateS[\agentA,\actionStateT]}$.

              \paragraph{atoms} By construction $\actionPrecondition[\prime](\proxyStateS[\agentA,\actionStateT]) = \actionPrecondition[\prime](\actionStateS[\agentA,\actionStateT])$.

              \paragraph{forth-$i$-$\agentB$} Suppose that $0 < i \leq n - 1$. Let $\actionStateU \in \proxyStateS[\agentA,\actionStateT] \actionAccessibilityAgent[\prime]{\agentA}$. 

              Suppose that $\agentB = \agentA$.
              By construction $\proxyStateS[\agentA,\actionStateT] \actionAccessibilityAgent[\prime]{\agentA} = \{\proxyStateS[\agentA,\actionStateV] \mid \actionStateV \in \actionStateT \actionAccessibilityAgent{\agentA}\} \cup \{\actionStateS[\prime]\}$. 
              Suppose that $\actionStateU \in \{\proxyStateS[\agentA,\actionStateV] \mid \actionStateV \in \actionStateT \actionAccessibilityAgent{\agentA}\}$.
              Then there exists $\actionStateV \in \actionStateT \actionAccessibilityAgent{\agentA}$ such that $\actionStateU = \proxyStateS[\agentA,\actionStateV]$.
              By the outer induction hypothesis $\pointedActionModel[\agentA,\actionStateT]{\actionStateS[\agentA,\actionStateT]} \bisimilar_{(n - 1)} \pointedActionModel{\actionStateT}$.
              As $\actionStateV \in \actionStateT \actionAccessibilityAgent{\agentA}$ then by {\bf back-$(n-1)$-$\agentA$}
              there exists $\actionStateW \in \actionStateS[\agentA,\actionStateT] \actionAccessibilityAgent[\agentA,\actionStateT]{\agentA} \subseteq \actionStateS[\agentA,\actionStateT] \actionAccessibilityAgent[\prime]{\agentA}$
              such that $\pointedActionModel[\agentA,\actionStateT]{\actionStateW} \bisimilar_{(n - 2)} \pointedActionModel{\actionStateV}$.
              Then by the inner induction hypothesis this implies $\pointedActionModel[\prime]{\actionStateW} \bisimilar_{(i - 1)} \pointedActionModel{\actionStateV}$.
              By the inner and outer induction hypothesis $\pointedActionModel[\prime]{\proxyStateS[\agentA,\actionStateV]} \bisimilar_{(i - 1)} \pointedActionModel{\actionStateV}$.
              Therefore by transitivity we have that $\pointedActionModel[\prime]{\proxyStateS[\agentA,\actionStateV]} \bisimilar_{(i - 1)} \pointedActionModel[\prime]{\actionStateW}$.
              Suppose that $\actionStateU = \actionStateS[\prime]$. 
              Then from the inner induction hypothesis $\pointedActionModel[\prime]{\actionStateS[\prime]} \bisimilar_{(i - 1)} \pointedActionModel[\prime]{\proxyStateS[\agentA,\actionStateS]}$ 
              and we can proceed using the same reasoning as in the case where $\actionStateU = \proxyStateS[\agentA,\actionStateS] \in \{\proxyStateS[\agentA,\actionStateV] \mid \actionStateV \in \actionStateT \actionAccessibilityAgent{\agentA}\}$.

              Suppose that $\agentB \neq \agentA$.
              By construction $\proxyStateS[\agentB,\actionStateT] \actionAccessibilityAgent[\prime]{\agentB} = \actionStateS[\agentA,\actionStateT] \actionAccessibilityAgent[\agentA,\actionStateT]{\agentB} \cup \{\proxyStateS[\agentB,\actionStateT]\}$. 
              Suppose that $\actionStateU = \proxyStateS[\agentB,\actionStateT]$.
              By construction $\actionStateS[\agentA,\actionStateT] \in \actionStateS[\agentA,\actionStateT] \actionAccessibilityAgent[\prime]{\agentB}$
              and by the induction hypothesis $\pointedActionModel[\prime]{\proxyStateS[\agentB,\actionStateT]} \bisimilar_{(i - 1)} \pointedActionModel[\prime]{\actionStateS[\agentA,\actionStateT]}$.
              Suppose that $\actionStateU \in \actionStateS[\agentA,\actionStateT] \actionAccessibilityAgent[\agentA,\actionStateT]{\agentB} \subseteq \actionStateS[\agentA,\actionStateT] \actionAccessibilityAgent[\prime]{\agentB}$.
              Then we trivially have that $\pointedActionModel[\prime]{\actionStateU} \bisimilar \pointedActionModel[\prime]{\actionStateU}$.

              \paragraph{back-$i$-$\agentB$} Follows similar reasoning to {\bf forth-$i$-$\agentB$}.

          \item 
              For every $\agentA \in \agents$, $\actionStateT \in \actionStateS \actionAccessibilityAgent{\agentA}$, $\actionStateU \in \actionStates[\agentA,\actionStateT]$, $\actionStateV \in \actionStates$: if $\pointedActionModel[\agentA,\actionStateT]{\actionStateU} \bisimilar_i \pointedActionModel{\actionStateV}$ then $\pointedActionModel[\prime]{\actionStateU} \bisimilar_i \pointedActionModel{\actionStateV}$.

              Suppose that $\pointedActionModel[\agentA,\actionStateT]{\actionStateU} \bisimilar_i \pointedActionModel{\actionStateV}$. 

              \paragraph{atoms} As $\pointedActionModel[\agentA,\actionStateT]{\actionStateU} \bisimilar_i \pointedActionModel{\actionStateV}$ 
              then $\proves \actionPrecondition[\agentA,\actionStateT](\actionStateU) \iff \actionPrecondition(\actionStateV)$. 
              By construction $\actionPrecondition[\prime](\actionStateU) = \actionPrecondition[\agentA,\actionStateT](\actionStateU)$ 
              and therefore $\proves \actionPrecondition[\prime](\actionStateU) \iff \actionPrecondition(\actionStateV)$.

              \paragraph{forth-$i$-$\agentB$} Suppose that $0 < i \leq n - 1$.
              Let $\actionStateW \in \actionStateU \actionAccessibilityAgent[\prime]{\agentB}$.

              Suppose that $\actionStateU \neq \actionStateS[\agentA,\actionStateT]$ or $\agentB = \agentA$.
              By construction $\actionStateU \actionAccessibilityAgent[\prime]{\agentA} = \actionStateU \actionAccessibilityAgent[\agentA,\actionStateT]{\agentA}$
              and so $\actionStateW \in \actionStateU \actionAccessibilityAgent[\agentA,\actionStateT]{\agentA}$. 
              As $\actionStateW \in \actionStateU \actionAccessibilityAgent[\agentA,\actionStateT]{\agentA}$
              then by {\bf forth-$i$-$\agentB$} there exists $\actionStateX \in \actionStateV \actionAccessibilityAgent{\agentB}$ such that
              $\pointedActionModel[\agentA,\actionStateT]{\actionStateW} \bisimilar_{(i - 1)} \pointedActionModel{\actionStateX}$.
              By the induction hypothesis $\pointedActionModel[\prime]{\actionStateW} \bisimilar_{(i - 1)} \pointedActionModel{\actionStateX}$.

              Suppose that $\actionStateU = \actionStateS[\agentA,\actionStateT]$ and $\agentB \neq \agentA$. 
              By construction $\actionStateS[\agentA,\actionStateT] \actionAccessibilityAgent[\prime]{\agentA} = \actionStateS[\agentA,\actionStateT] \actionAccessibilityAgent[\agentA,\actionStateT]{\agentA} \cup \{\proxyStateS[\agentA,\actionStateT]\}$. 
              Suppose that $\actionStateW \in \actionStateS[\agentA,\actionStateT] \actionAccessibilityAgent[\agentA,\actionStateT]{\agentA}$. 
              We proceed using the same reasoning as above, where $\actionStateW \in \actionStateU \actionAccessibilityAgent[\agentA,\actionStateT]{\agentA}$. 
              Suppose that $\actionStateW = \proxyStateS[\agentA,\actionStateT]$.
              By the induction hypothesis $\pointedActionModel[\prime]{\proxyStateS[\agentA,\actionStateT]} \bisimilar_{(i - 1)} \pointedActionModel[\prime]{\actionStateS[\agentA,\actionStateT]}$
              and we proceed using the same reasoning above, where $\actionStateW = \actionStateS[\agentA,\actionStateT] \in \actionStateS[\agentA,\actionStateT] \actionAccessibilityAgent[\agentA,\actionStateT]{\agentA}$.

              \paragraph{back-$i$-$\agentB$} Follows similar reasoning to {\bf forth-$i$-$\agentB$}.
      \end{enumerate}

      Therefore for every $\agentA \in \agents$, 
      $\actionStateT \in \actionStateS \actionAccessibilityAgent{\agentA}$
      we have that $\pointedActionModel[\prime]{\actionStateS[\prime]} \bisimilar_{(n - 1)} \pointedActionModel{\actionStateS}$
      and $\pointedActionModel[\prime]{\proxyStateS[\agentA,\actionStateT]} \bisimilar_{(n - 1)} \pointedActionModel{\actionStateT}$.
      We can now show that $\pointedActionModel{\actionStateS[\prime]} \bisimilar_n \pointedActionModel{\actionStateS}$ 
      by using the same reasoning as the proof for Proposition~\ref{afl-k-correspondence}, 
      using the $(n-1)$-bisimilar $\pointedActionModel[\prime]{\proxyStateS[\agentA,\actionStateT]}$ 
      in place of corresponding $\pointedActionModel[\prime]{\actionStateS[\actionStateT]}$ states.
  \end{proof}

  \begin{corollary}
      Let $\pointedActionModel{\actionStateS} \in \classAM_\classS$.
      Then for every $\phi \in \langAml$
      there exists $\alpha \in \langAflAct$ 
      such that $\entails_\logicAmlS{} \allacts{\pointedActionModel{\actionStateS}} \phi \iff \allacts{\tau(\alpha)} \phi$.
  \end{corollary}

  \begin{corollary}
      Let $\phi \in \langAml$. 
      Then there exists $\phi' \in \langAfl$ 
      such that for every $\pointedModel{\stateS} \in \classS$: 
      $\pointedModel{\stateS} \entails_\logicAmlS{} \phi$ if and only if
      $\pointedModel{\stateS} \entails_\logicAflS{} \phi'$.
  \end{corollary}

  \section{Synthesis}\label{synthesis}

  In the following subsections we give a computational method for synthesising
  action formulae to achieve epistemic goals, whenever those goals are
  achievable.  We note that the notion of when an epistemic goal is achievable
  is captured by the refinement quantifiers of refinement modal
  logic~\cite{vanditmarsch2009,bozzelli2012a}, which are also included in the
  arbitrary action formula logic, and so in this section we will refer to the full
  arbitrary action formula logic, keeping in mind the correspondence with 
  arbitrary action model logic mentioned in Section~\ref{semantics}.

  \subsection{\classK{}}

  \begin{proposition}\label{afl-k-synthesis}
      For every $\phi \in \langAfl$ there exists $\alpha \in \langAflAct$ such that $\proves \allacts{\alpha} \phi$ and $\proves \somerefs \phi \implies \someacts{\alpha} \phi$.
  \end{proposition}

  \begin{proof}
      Without loss of generality we assume that $\phi$ is in disjunctive normal
      form. We proceed by induction on the structure of $\phi$.

      Suppose that $\phi = \psi \lor \chi$. By the induction hypothesis 
      there exists $\alpha^\psi, \alpha^\chi \in \langAflAct$ 
      such that $\proves \allacts{\alpha^\psi} \psi$, 
      $\proves \somerefs \psi \implies \someacts{\alpha^\psi} \psi$,
      $\proves \allacts{\alpha^\chi} \chi$ and
      $\proves \somerefs \chi \implies \someacts{\alpha^\chi} \chi$.
      Let $\alpha = \alpha^\psi \choice \alpha^\chi$.
      Then:
      \begin{eqnarray}
          &\proves& \allacts{\alpha^\psi} (\psi \lor \chi) \land \allacts{\alpha^\chi} (\psi \lor \chi)\label{afl-k-synthesis-or-1}\\
          &\proves& \allacts{\alpha^\psi \choice \alpha^\chi} (\psi \lor \chi)\label{afl-k-synthesis-or-2}
      \end{eqnarray}
      (\ref{afl-k-synthesis-or-1}) follows from the induction hypothesis and
      (\ref{afl-k-synthesis-or-2}) follows from {\bf LU}.

      Further:
      \begin{eqnarray}
          &\proves& (\somerefs \psi \lor \somerefs \chi) \implies (\someacts{\alpha^\psi} (\psi \lor \chi) \lor \someacts{\alpha^\chi} (\psi \lor \chi))\label{afl-k-synthesis-or-3}\\
          &\proves& (\somerefs \psi \lor \somerefs \chi) \implies \someacts{\alpha^\psi \choice \alpha^\chi} (\psi \lor \chi)\label{afl-k-synthesis-or-4}\\
          &\proves& \somerefs (\psi \lor \chi) \implies \someacts{\alpha^\psi \choice \alpha^\chi} (\psi \lor \chi)\label{afl-k-synthesis-or-5}
      \end{eqnarray}
      (\ref{afl-k-synthesis-or-3}) follows from the induction hypothesis,
      (\ref{afl-k-synthesis-or-4}) follows from {\bf LU} and
      (\ref{afl-k-synthesis-or-5}) follows from {\bf R}.

      Suppose that $\phi = \pi \land \bigwedge_{\agentB \in \agentsB \subseteq \agents} \covers_\agentB \Gamma_\agentB$.
      By the induction hypothesis for every $\agentB \in \agentsB$, $\gamma \in \Gamma_\agentB$
      there exists $\alpha^\gamma \in \langAflAct$ such that 
      $\proves \allacts{\alpha^\gamma} \gamma$ and 
      $\proves \somerefs \gamma \implies \someacts{\alpha^\gamma} \gamma$.
      Let $\alpha = \test{\somerefs \phi} \compose \bigcompose_{\agentB \in \agentsB} \learns_\agentB (\bigchoice_{\gamma \in \Gamma_\agentB} \alpha^\gamma)$.

      Then for every $\agentB \in \agentsB$: 
      \begin{eqnarray}
          &\proves& \allacts{\bigchoice_{\gamma \in \Gamma_\agentB} \alpha^\gamma} \bigvee_{\gamma \in \Gamma} \gamma\label{afl-k-synthesis-covers-1}\\
          &\proves& \necessary_\agentB \allacts{\bigchoice_{\gamma \in \Gamma_\agentB} \alpha^\gamma} \bigvee_{\gamma \in \Gamma} \gamma\label{afl-k-synthesis-covers-2}\\
          &\proves& \allacts{\learns_\agentB (\bigchoice_{\gamma \in \Gamma_\agentB} \alpha^\gamma)} \necessary_\agentB \bigvee_{\gamma \in \Gamma} \gamma\label{afl-k-synthesis-covers-3}\\
          &\proves& \allacts{\bigcompose_{\agentC \in \agentsB} \learns_\agentC (\bigchoice_{\gamma \in \Gamma_\agentC} \alpha^\gamma)} \necessary_\agentB \bigvee_{\gamma \in \Gamma} \gamma\label{afl-k-synthesis-covers-4}\\
          &\proves& \allacts{\test{\somerefs \phi}} \allacts{\bigcompose_{\agentC \in \agentsB} \learns_\agentC (\bigchoice_{\gamma \in \Gamma_\agentC} \alpha^\gamma)} \necessary_\agentB \bigvee_{\gamma \in \Gamma} \gamma\label{afl-k-synthesis-covers-5}\\
          &\proves& \allacts{\test{\somerefs \phi} \compose \bigcompose_{\agentC \in \agentsB} \learns_\agentC (\bigchoice_{\gamma \in \Gamma_\agentC} \alpha^\gamma)} \necessary_\agentB \bigvee_{\gamma \in \Gamma} \gamma\label{afl-k-synthesis-covers-6}
      \end{eqnarray}
      (\ref{afl-k-synthesis-covers-1}) follows from the induction hypothesis and {\bf LU},
      (\ref{afl-k-synthesis-covers-2}) follows from {\bf NecK},
      (\ref{afl-k-synthesis-covers-3}) follows from {\bf LK1},
      (\ref{afl-k-synthesis-covers-4}) follows from {\bf LK2} and {\bf LS},
      (\ref{afl-k-synthesis-covers-5}) follows from {\bf NecL} and
      (\ref{afl-k-synthesis-covers-6}) follows from {\bf LS}.

      Further:
      \begin{eqnarray}
          &\proves& \somerefs \phi \implies \bigwedge_{\agentB \in \agentsB, \gamma \in \Gamma_\agentB} \possible_\agentB \somerefs \gamma\label{afl-k-synthesis-covers-7}\\
          &\proves& \somerefs \phi \implies \bigwedge_{\agentB \in \agentsB, \gamma \in \Gamma_\agentB} \possible_\agentB \someacts{\alpha^{\gamma}} \gamma\label{afl-k-synthesis-covers-8}\\
          &\proves& \somerefs \phi \implies \bigwedge_{\agentB \in \agentsB, \gamma \in \Gamma_\agentB} \possible_\agentB \someacts{\bigchoice_{\gamma' \in \Gamma_\agentB} \alpha^{\gamma'}} \gamma\label{afl-k-synthesis-covers-9}\\
          &\proves& \somerefs \phi \implies \someacts{\bigcompose_{\agentC \in \agentsB} \learns_\agentC (\bigchoice_{\gamma \in \Gamma_\agentC} \alpha^\gamma)} \bigwedge_{\agentB \in \agentsB, \gamma \in \Gamma_\agentB} \possible_\agentB \gamma\label{afl-k-synthesis-covers-10}\\
          &\proves& \somerefs \phi \implies \someacts{\test{\somerefs \phi} \compose \bigcompose_{\agentC \in \agentsB} \learns_\agentC (\bigchoice_{\gamma \in \Gamma_\agentC} \alpha^\gamma)} \bigwedge_{\agentB \in \agentsB, \gamma \in \Gamma_\agentB} \possible_\agentB \gamma\label{afl-k-synthesis-covers-11}\\
          &\proves& \allacts{\test{\somerefs \phi} \compose \bigcompose_{\agentC \in \agentsB} \learns_\agentC (\bigchoice_{\gamma \in \Gamma_\agentC} \alpha^\gamma)} \bigwedge_{\agentB \in \agentsB, \gamma \in \Gamma_\agentB} \possible_\agentB \gamma\label{afl-k-synthesis-covers-12}\\
          &\proves& \allacts{\test{\somerefs \phi} \compose \bigcompose_{\agentC \in \agentsB} \learns_\agentC (\bigchoice_{\gamma \in \Gamma_\agentC} \alpha^\gamma)} (\pi \land \bigwedge_{\agentB \in \agentsB} \covers_\agentB \Gamma_\agentB)\label{afl-k-synthesis-covers-13}
      \end{eqnarray}
      (\ref{afl-k-synthesis-covers-7}) follows from {\bf RK},
      (\ref{afl-k-synthesis-covers-8}) follows from the induction hypothesis,
      (\ref{afl-k-synthesis-covers-9}) follows from {\bf LU},
      (\ref{afl-k-synthesis-covers-10}) follows from {\bf LK1}, {\bf LK2} and {\bf LS},
      (\ref{afl-k-synthesis-covers-11}) and (\ref{afl-k-synthesis-covers-12}) follow from {\bf LT}, and
      (\ref{afl-k-synthesis-covers-13}) follows from (\ref{afl-k-synthesis-covers-6}), {\bf RP} {\bf LC} and the definition of the cover operator.

      Therefore $\proves \allacts{\alpha} \phi$.

      Finally:
      \begin{eqnarray}
      &\proves& \someacts{\bigcompose_{\agentC \in \agentsB} \learns_\agentC (\bigchoice_{\gamma \in \Gamma_\agentC} \alpha^\gamma)} \top \iff \top\label{afl-k-synthesis-covers-14}\\
      &\proves& \someacts{\test{\somerefs \phi} \compose \bigcompose_{\agentC \in \agentsB} \learns_\agentC (\bigchoice_{\gamma \in \Gamma_\agentC} \alpha^\gamma)} \top \iff \somerefs \phi\label{afl-k-synthesis-covers-15}\\
      &\proves& \somerefs \phi \implies \someacts{\alpha} \top\label{afl-k-synthesis-covers-16}\\
      &\proves& \somerefs \phi \implies \someacts{\alpha} \phi\label{afl-k-synthesis-covers-17}
      \end{eqnarray}
      (\ref{afl-k-synthesis-covers-14}) follows from {\bf LS} and {\bf LP},
      (\ref{afl-k-synthesis-covers-15}) follows from {\bf LS} and {\bf LT},
      (\ref{afl-k-synthesis-covers-16}) follows from (\ref{afl-k-synthesis-covers-15}),
      (\ref{afl-k-synthesis-covers-17}) follows from (\ref{afl-k-synthesis-covers-13}) and (\ref{afl-k-synthesis-covers-16}),

      Therefore $\proves \somerefs \phi \implies \someacts{\alpha} \phi$.
  \end{proof}

  \begin{corollary}
      For every $\pointedModel{\stateS} \in \classK$ and $\phi \in \langAaml$: 
      $\pointedModel{\stateS} \entails \somerefs \phi$ if and only if 
      there exists $\pointedActionModel{\actionStateS} \in \classAM$ 
      such that $\pointedModel{\stateS} \entails \someacts{\pointedActionModel{\actionStateS}} \phi$.
  \end{corollary}

  \subsection{\classKFF{}}

  \begin{proposition}\label{afl-kff-synthesis}
      For every $\phi \in \langAfl$ there exists $\alpha \in \langAflAct$ such that $\proves \allacts{\alpha} \phi$ and $\proves \somerefs \phi \implies \someacts{\alpha} \phi$.
  \end{proposition} 

  \begin{proof}
      Without loss of generality we assume that $\phi$ is in alternating
      disjunctive normal form. We use the same reasoning as in the proof of
      Proposition~\ref{afl-k-synthesis}, substituting \axiomAflKFF{} axioms for
      the corresponding \axiomAflK{} axioms, noting that the alternating
      disjunctive normal form gives the $(\agents \setminus
      \{\agentA\})$-restricted properties required for {\bf LK1} and the
      \axiomRmlKFF{} axioms {\bf RK45}, {\bf RComm} and  {\bf RDist} to be
      applicable.
  \end{proof}

  \begin{corollary}
      For every $\pointedModel{\stateS} \in \classKFF$ and $\phi \in \langAaml$: 
      $\pointedModel{\stateS} \entails \somerefs \phi$ if and only if 
      there exists $\pointedActionModel{\actionStateS} \in \classAM_\classKFF$ 
      such that $\pointedModel{\stateS} \entails \someacts{\pointedActionModel{\actionStateS}} \phi$.
  \end{corollary}

  \subsection{\classS{}}

  \begin{proposition}\label{afl-s-synthesis}
      For every $\phi \in \langAfl$ there exists $\alpha \in \langAflAct$ such that $\proves \allacts{\alpha} \phi$ and $\proves \somerefs \phi \implies \someacts{\alpha} \phi$.
  \end{proposition}

  \begin{proof}
      Without loss of generality, assume that $\phi$ is a disjunction of
      explicit formulae. We proceed by induction on the structure of $\phi$.

      Suppose that $\phi = \psi \lor \chi$. We use the same reasoning as in the
      proof of Proposition~\ref{afl-k-synthesis}.

      Suppose that $\phi = \pi \land \gamma^0 \land \bigwedge_{\agentA \in \agents} \covers_\agentA \Gamma_\agentA$ is an explicit formula.
      By the induction hypothesis for every $\agentA \in \agents$, $\gamma \in \Gamma_\agentA$ 
      there exists $\alpha^{\agentA,\gamma} \in \langAflAct$
      such that $\proves \allacts{\alpha^{\agentA,\gamma}} \gamma$ and $\proves \somerefs \gamma \implies \someacts{\alpha^{\agentA,\gamma}} \gamma$,
      where $\tau(\alpha^{\agentA,\gamma}) = \pointedActionModel[\agentA,\gamma]{\actionStateS[\agentA,\gamma]} = \pointedActionModelTuple[\agentA,\gamma]{\actionStateS[\agentA,\gamma]}$.

      Let $\alpha = \test{\somerefs \gamma^0} \compose \bigcompose_{\agentA \in \agents} \learns_\agentA (\test{\top}, \bigchoice_{\gamma \in \Gamma_\agentA} \alpha^{\agentA,\gamma})$.
      Then from Lemmas~\ref{afl-s-construction-test} and~\ref{afl-s-construction-learning}: $\tau(\alpha) \bisimilar \pointedActionModel{\actionStateS} = \pointedActionModelTuple{\actionStateS}$ where:
      \begin{eqnarray*}
          \actionStates &=& \bigcup_{\agentA \in \agents, \gamma \in \Gamma_\agentA} \actionStates[\agentA,\gamma] \cup \{\proxyStateS[\agentA,\gamma] \mid \agentA \in \agents, \gamma \in \Gamma_\agentA\} \cup \{\actionStateS\}\\
          \actionAccessibilityAgent{\agentA} &=& \bigcup_{\agentB \in \agents, \gamma \in \Gamma_\agentB} \actionAccessibilityAgent[\agentB,\gamma]{\agentA} \cup (\{\actionStateS\} \cup \{\proxyStateS[\agentA,\gamma] \mid \gamma \in \Gamma_\agentA\})^2 \cup \bigcup_{\agentB \in \agents \setminus \{\agentA\}, \gamma \in \Gamma_\agentB} (\{\proxyStateS[\agentB,\gamma]\} \cup \actionStateS[\agentB,\gamma] \actionAccessibilityAgent[\agentB,\gamma]{\agentA})^2 \text{ for } \agentA \in \agents\\
          \actionPrecondition &=& \bigcup_{\agentA \in \agents, \gamma \in \Gamma_\agentA} \actionPrecondition[\agentA,\gamma] \cup \{(\proxyStateS[\agentA,\gamma], \actionPrecondition[\agentA,\gamma](\actionStateS[\agentA<,\gamma])) \mid \agentA \in \agents, \gamma \in \gamma_\agentA\} \cup \{(\actionStateS, \somerefs \gamma^0)\}
      \end{eqnarray*}

      Let $\Psi = \{\psi \leq \gamma \mid \agentA \in \agents, \gamma \in \Gamma_\agentA\}$. We need to show for every $\psi \in \Psi$:

      \begin{enumerate}
          \item For every $\agentA \in \agents$: $\proves \allacts{\pointedActionModel{\actionStateS}} \psi \iff \allacts{\pointedActionModel{\actionStateS[\agentA,\gamma^0]}} \psi$.
          \item For every $\agentA \in \agents$, $\gamma \in \Gamma_\agentA$: $\proves \allacts{\pointedActionModel{\proxyStateS[\agentA,\gamma]}} \psi \iff \allacts{\pointedActionModel{\actionStateS[\agentA,\gamma]}} \psi$.
          \item For every $\agentA \in \agents$, $\gamma \in \Gamma_\agentA$, $\actionStateU \in \actionStates[\agentA,\gamma]$: $\proves \allacts{\pointedActionModel{\actionStateU}} \psi \iff \allacts{\pointedActionModel[\agentA,\gamma]{\actionStateU}} \psi$.
      \end{enumerate}

      We proceed by induction on $\psi$.

      \begin{enumerate}
          \item For every $\agentA \in \agents$: $\proves \allacts{\pointedActionModel{\actionStateS}} \psi \iff \allacts{\pointedActionModel{\actionStateS[\agentA,\gamma^0]}} \psi$.

              Suppose that $\psi = \atomP$ where $\atomP \in \atoms$. 
              This follows trivially from {\bf AP}.

              Suppose that $\psi = \neg \chi$ or that $\psi = \chi_1 \land \chi_2$. These cases follow trivially from the induction hypothesis.

              Suppose that $\psi = \necessary[\agentA] \chi$.
              By construction $\actionStateS \actionAccessibilityAgent{\agentA} = \proxyStateS[\agentA,\gamma^0] \actionAccessibilityAgent{\agentA}$ 
              and $\actionPrecondition(\actionStateS) = \actionPrecondition(\proxyStateS[\agentA,\gamma^0])$ 
              and so $\proves \allacts{\pointedActionModel{\actionStateS}} \necessary[\agentA] \chi \iff \allacts{\pointedActionModel{\proxyStateS[\agentA,\gamma^0]}} \necessary[\agentA] \chi$
              follows from {\bf AK} trivially.

              Suppose that $\psi = \necessary[\agentB] \chi$ where $\agentB \neq \agentA$. 
              By construction $\actionStateS \actionAccessibilityAgent{\agentB} = \{\actionStateS\} \cup \actionStateS[\agentB,\gamma^0] \actionAccessibilityAgent{\agentB}$ 
              and $\proxyStateS[\agentA,\gamma^0] \actionAccessibilityAgent{\agentB} = \{\proxyStateS[\agentA,\gamma^0]\} \cup \actionStateS[\agentA,\gamma^0] \actionAccessibilityAgent[\agentA,\gamma^0]{\agentB}$.
              As $\phi$ is an explicit formula and $\necessary[\agentB] \chi \in \Psi$ 
              then either $\proves \gamma^0 \implies \necessary[\agentB] \chi$ 
              or $\proves \gamma^0 \implies \neg \necessary[\agentB] \chi$.
              Suppose that $\proves \gamma^0 \implies \necessary[\agentB] \chi$.
              Then for every $\gamma \in \Gamma_\agentB$ we have $\proves \gamma \implies \necessary[\agentB] \chi$.
              By the outer induction hypothesis $\proves \allacts{\pointedActionModel[\agentB,\gamma]{\actionStateS[\agentB,\gamma]}} \gamma$
              and so $\proves \allacts{\pointedActionModel[\agentB,\gamma]{\actionStateS[\agentB,\gamma]}} \chi$.
              By the inner induction hypothesis $\proves \allacts{\pointedActionModel{\actionStateS[\agentB,\gamma]}} \chi$.
              As $\gamma^0 \in \Gamma_\agentB$ then $\proves \allacts{\pointedActionModel{\actionStateS[\agentB,\gamma^0]}} \chi$
              and so by the inner induction hypothesis $\proves \allacts{\pointedActionModel{\actionStateS}} \chi$.
              So $\proves \allacts{\pointedActionModel{\actionStateS \actionAccessibilityAgent{\agentB}}} \chi$ 
              and therefore $\proves \allacts{\pointedActionModel{\actionStateS}} \necessary[\agentB] \chi$ follows from {\bf AK}.
              By the outer induction hypothesis $\proves \allacts{\pointedActionModel[\agentA,\gamma^0]{\actionStateS[\agentA,\gamma^0]}} \gamma^0$
              and so $\proves \allacts{\pointedActionModel[\agentA,\gamma^0]{\actionStateS[\agentA,\gamma^0]}} \necessary[\agentB] \chi$.
              From {\bf AK} we have $\proves \somerefs{\gamma^0} \implies \necessary[\agentB] \allacts{\pointedActionModel[\agentA,\gamma^0]{\actionStateS[\agentA,\gamma^0] \actionAccessibilityAgent[\agentA,\gamma^0]{\agentB}}} \chi$.
              By the inner induction hypothesis $\proves \allacts{\pointedActionModel[\agentA,\gamma^0]{\actionStateS[\agentA,\gamma^0] \actionAccessibilityAgent[\agentA,\gamma^0]{\agentB}}} \chi \iff \allacts{\pointedActionModel{\actionStateS[\agentA,\gamma^0] \actionAccessibilityAgent[\agentA,\gamma^0]{\agentB}}} \chi$.
              and as $\proves \allacts{\pointedActionModel{\actionStateS}} \chi$
              then $\proves \allacts{\pointedActionModel{\proxyStateS[\agentA,\gamma^0]}} \chi$.
              So we have $\proves \allacts{\pointedActionModel[\agentA,\gamma^0]{\actionStateS[\agentA,\gamma^0] \actionAccessibilityAgent[\agentA,\gamma^0]{\agentB}}} \chi \iff \allacts{\pointedActionModel{\proxyStateS[\agentA,\gamma^0] \actionAccessibilityAgent{\agentB}}} \chi$
              and $\proves \somerefs{\gamma^0} \implies \necessary[\agentB] \allacts{\pointedActionModel{\proxyStateS[\agentA,\gamma^0] \actionAccessibilityAgent{\agentB}}} \chi$
              and so $\proves \allacts{\pointedActionModel{\proxyStateS[\agentA,\gamma^0]}} \necessary[\agentB]$ follows from {\bf AK}.
              Therefore $\proves \allacts{\pointedActionModel{\actionStateS}} \necessary[\agentB] \chi \iff \allacts{\pointedActionModel{\proxyStateS[\agentA,\gamma^0]}} \necessary[\agentB] \chi$.
              Suppose that $\proves \gamma^0 \implies \neg \necessary[\agentB] \chi$.
              A dual argument can be used to show that $\proves \neg \allacts{\pointedActionModel{\actionStateS}} \necessary[\agentB] \chi$ 
              and $\proves \neg \allacts{\pointedActionModel{\proxyStateS[\agentA,\gamma^0]}} \necessary[\agentB] \chi$
              and therefore $\proves \allacts{\pointedActionModel{\actionStateS}} \necessary[\agentB] \chi \iff \allacts{\pointedActionModel{\proxyStateS[\agentA,\gamma^0]}} \necessary[\agentB] \chi$.

          \item For every $\agentA \in \agents$, $\gamma \in \Gamma_\agentA$: $\proves \allacts{\pointedActionModel{\proxyStateS[\agentA,\gamma]}} \psi \iff \allacts{\pointedActionModel{\actionStateS[\agentA,\gamma]}} \psi$.

              Suppose that $\psi = \atomP$ where $\atomP \in \atoms$. 
              This follows trivially from {\bf AP}.

              Suppose that $\psi = \neg \chi$ or that $\psi = \chi_1 \land \chi_2$. These cases follow trivially from the induction hypothesis.

              Suppose that $\psi = \necessary[\agentA] \chi$.
              By construction $\proxyStateS[\agentA,\gamma] \actionAccessibilityAgent{\agentA} = \{\actionStateS\} \cup \{\proxyStateS[\agentA,\gamma] \mid \delta \in \Gamma_\agentA\}$ 
              and $\actionStateS[\agentA,\gamma] \actionAccessibilityAgent{\agentA} = \actionStateS[\agentA,\gamma] \actionAccessibilityAgent[\agentA,\gamma]{\agentA}$.
              As $\phi$ is an explicit formula and $\necessary[\agentA] \chi \in \Psi$ 
              then either $\proves \gamma \implies \necessary[\agentA] \chi$ 
              or $\proves \gamma \implies \neg \necessary[\agentA] \chi$.
              Suppose that $\proves \gamma \implies \necessary[\agentA] \chi$.
              Then for every $\delta \in \Gamma_\agentA$ 
              we have $\proves \delta \implies \necessary[\agentA] \chi$.
              By the outer induction hypothesis $\proves \allacts{\pointedActionModel[\agentA,\delta]{\actionStateS[\agentA,\delta]}} \delta$ 
              and so $\proves \allacts{\pointedActionModel[\agentA,\delta]{\actionStateS[\agentA,\delta]}} \chi$.
              By the inner induction hypothesis $\proves \allacts{\pointedActionModel{\actionStateS[\agentA,\delta]}} \chi$
              and $\proves \allacts{\pointedActionModel{\proxyStateS[\agentA,\delta]}} \chi$.
              As $\gamma^0 \in \Gamma_\agentA$ then $\proves \allacts{\pointedActionModel{\proxyStateS[\agentA,\gamma^0]}} \chi$ 
              and by the inner induction hypothesis $\proves \allacts{\pointedActionModel{\actionStateS}} \chi$.
              So $\proves \allacts{\pointedActionModel{\proxyStateS[\agentA,\gamma] \actionAccessibilityAgent{\agentA}}} \chi$ 
              and therefore $\proves \allacts{\pointedActionModel{\proxyStateS[\agentA,\gamma]}} \necessary[\agentA] \chi$
              follows from {\bf AK}.
              By the outer induction hypothesis $\proves \allacts{\pointedActionModel[\agentA,\gamma]{\actionStateS[\agentA,\gamma]}} \gamma$ 
              and so $\proves \allacts{\pointedActionModel[\agentA,\gamma]{\actionStateS[\agentA,\gamma]}} \necessary[\agentA] \chi$.
              From {\bf AK} we have $\proves \somerefs \gamma \implies \necessary[\agentA] \allacts{\pointedActionModel[\agentA,\gamma]{\actionStateS[\agentA,\gamma] \actionAccessibilityAgent[\agentA,\gamma]{\agentA}}} \chi$.
              By the inner induction hypothesis $\proves \allacts{\pointedActionModel[\agentA,\gamma]{\actionStateS[\agentA,\gamma] \actionAccessibilityAgent[\agentA,\gamma]{\agentA}}} \chi \iff \allacts{\pointedActionModel{\actionStateS[\agentA,\gamma] \actionAccessibilityAgent[\agentA,\gamma]{\agentA}}} \chi$ 
              so $\proves \somerefs \gamma \implies \necessary[\agentA] \allacts{\pointedActionModel{\actionStateS[\agentA,\gamma] \actionAccessibilityAgent{\agentA}}} \chi$
              and so $\proves \allacts{\pointedActionModel{\actionStateS[\agentA,\gamma]}} \necessary[\agentA] \chi$
              follows from {\bf AK}.

              Suppose that $\proves \gamma \implies \neg \necessary[\agentB] \chi$ where $\agentB \neq \agentA$.
              Therefore $\proves \allacts{\pointedActionModel{\proxyStateS[\agentA,\gamma]}} \necessary[\agentA] \chi \iff \allacts{\pointedActionModel{\actionStateS[\agentA,\gamma]}} \necessary[\agentA] \chi$.
              A dual argument can be used to show that $\proves \neg \allacts{\pointedActionModel{\proxyStateS[\agentA,\gamma]}} \necessary[\agentA] \chi$
              and $\proves \neg \allacts{\pointedActionModel{\actionStateS[\agentA,\gamma]}} \necessary[\agentA] \chi$
              and therefore $\proves \allacts{\pointedActionModel{\proxyStateS[\agentA,\gamma]}} \necessary[\agentA] \chi \iff \allacts{\pointedActionModel{\actionStateS[\agentA,\gamma]}} \necessary[\agentA] \chi$.

              Suppose that $\psi = \necessary[\agentB] \chi$ where $\agentB \neq \agentA$.
              By construction $\proxyStateS[\agentA,\gamma] \actionAccessibilityAgent{\agentB} = \actionStateS[\agentA,\gamma] \actionAccessibilityAgent{\agentB}$
              and $\actionPrecondition(\proxyStateS[\agentA,\gamma]) = \actionPrecondition(\actionStateS[\agentA,\gamma])$ 
              and so $\proves \allacts{\pointedActionModel{\proxyStateS[\agentA,\gamma]}} \necessary[\agentB] \chi \iff \allacts{\pointedActionModel{\actionStateS[\agentA,\gamma]}} \necessary[\agentB] \chi$
              follows from {\bf AK} trivially.

          \item For every $\agentA \in \agents$, $\gamma \in \Gamma_\agentA$, $\actionStateU \in \actionStates[\agentA,\gamma]$: $\proves \allacts{\pointedActionModel{\actionStateU}} \psi \iff \allacts{\pointedActionModel[\agentA,\gamma]{\actionStateU}} \psi$.

              Suppose that $\psi = \atomP$ where $\atomP \in \atoms$. 
              This follows trivially from {\bf AP}.

              Suppose that $\psi = \neg \chi$ or that $\psi = \chi_1 \land \chi_2$. These cases follow trivially from the induction hypothesis.

              Suppose that $\psi = \necessary[\agentA] \chi$.
              By construction $\actionStateU \actionAccessibilityAgent{\agentA} = \actionStateU \actionAccessibilityAgent[\agentA,\gamma]{\agentA}$
              and $\actionPrecondition(\actionStateU) = \actionPrecondition[\agentA,\gamma](\actionStateU)$
              and so $\proves \allacts{\pointedActionModel{\actionStateU}} \necessary[\agentA] \chi \iff \allacts{\pointedActionModel[\agentA,\gamma]{\actionStateU}} \necessary[\agentA] \chi$
              follows from {\bf AK} and the induction hypothesis trivially.

              Suppose that $\psi = \necessary[\agentB] \chi$ where $\agentB \neq \agentA$.
              By construction $\actionStateU \actionAccessibilityAgent{\agentA} = \actionStateU \actionAccessibilityAgent[\agentA,\gamma]{\agentA}$
              or $\actionStateU \actionAccessibilityAgent{\agentA} = \{\proxyStateS[\agentA,\gamma]\} \cup \actionStateU \actionAccessibilityAgent[\agentA,\gamma]{\agentA}$
              and $\actionPrecondition(\actionStateU) = \actionPrecondition[\agentA,\gamma](\actionStateU)$
              and so $\proves \allacts{\pointedActionModel{\actionStateU}} \necessary[\agentB] \chi \iff \allacts{\pointedActionModel[\agentA,\gamma]{\actionStateU}} \necessary[\agentB] \chi$
              follows from {\bf AK} and the induction hypothesis trivially.
      \end{enumerate}

      Therefore for every $\agentA \in \agents$, $\gamma \in \Gamma_\agentA$ 
      we have that $\proves \allacts{\pointedActionModel{\actionStateS[\agentA,\gamma]}} \gamma$ 
      and $\proves \allacts{\pointedActionModel{\actionStateS}} \gamma^0$.
      Therefore for every $\agentA \in \agents$ 
      we have $\proves \allacts{\pointedActionModel{\actionStateS \actionAccessibilityAgent{\agentA}}} \bigvee_{\gamma \in \Gamma_\agentA} \gamma$
      and so from {\bf AK} we have that $\proves \allacts{\pointedActionModel{\actionStateS}} \necessary[\agentA] \bigvee_{\gamma \in \Gamma_\agentA} \gamma$.

      As $\phi$ is an explicit formula, from {\bf RDist}, {\bf RS5} and {\bf RComm} we have that
      $\somerefs \phi \implies \pi \land \bigwedge_{\agentA \in \agents, \gamma \in \Gamma_\agentA} \possible[\agentA] \somerefs \gamma$.
      By construction for every $\agentA \in \agents$, $\gamma \in \Gamma_\agentA$ 
      we have $\actionPrecondition(\proxyStateS[\agentA,\gamma]) = \somerefs \gamma$ 
      and from above we have $\proves \allacts{\pointedActionModel{\actionStateS[\agentA,\gamma]}} \gamma$ 
      therefore $\proves \somerefs \phi \implies \pi \land \bigwedge_{\agentA \in \agents, \gamma \in \Gamma_\agentA} \possible[\agentA] \someacts{\pointedActionModel{\proxyStateS[\agentA,\gamma]}} \gamma$.
      Therefore by {\bf AK} we have $\proves \somerefs \phi \implies \someacts{\pointedActionModel{\actionStateS}} (\pi \land \bigwedge_{\agentA \in \agents, \gamma \in \Gamma_\agentA} \possible[\agentA] \gamma)$.
      From above we have $\proves \allacts{\pointedActionModel{\actionStateS}} \necessary[\agentA] \bigvee_{\gamma \in \Gamma_\agentA} \gamma$
      and therefore $\proves \somerefs \phi \implies \allacts{\pointedActionModel{\actionStateS}} \phi$.
      As $\proves \phi \implies \gamma^0$ then $\proves \somerefs \phi \implies \somerefs \gamma^0$ 
      and so $\proves \somerefs \phi \implies \someacts{\pointedActionModel{\actionStateS}} \phi$.

      Let $\alpha' = \test{\somerefs \phi} \compose \alpha$.
      By {\bf LS} we have $\proves \allacts{\alpha'} \phi \iff \allacts{\test{\somerefs \phi}} \allacts{\alpha} \phi$.
      By {\bf LT} we have $\proves \allacts{\alpha'} \phi \iff (\somerefs \phi \implies \allacts{\alpha} \phi)$.
      From above we have $\proves \somerefs \phi \implies \allacts{\alpha} \phi$
      and therefore $\proves \allacts{\alpha'} \phi$.
      By {\bf LS} we have $\proves \someacts{\alpha'} \phi \iff \someacts{\test{\somerefs \phi}} \someacts{\alpha} \phi$.
      By {\bf LT} we have $\proves \someacts{\alpha'} \phi \iff (\somerefs \phi \land \someacts{\alpha} \phi)$.
      From above we have $\proves \somerefs \phi \implies \someacts{\alpha} \phi$
      and therefore $\proves \somerefs \phi \implies \someacts{\alpha'} \phi$.
  \end{proof}

  \begin{corollary}
      For every $\pointedModel{\stateS} \in \classS$ and $\phi \in \langAaml$: 
      $\pointedModel{\stateS} \entails \somerefs \phi$ if and only if 
      there exists $\pointedActionModel{\actionStateS} \in \classAM_\classS$ 
      such that $\pointedModel{\stateS} \entails \someacts{\pointedActionModel{\actionStateS}} \phi$.
  \end{corollary}

  \section{Related work}\label{related-work}

  Several other papers have addressed the problem of describing and reasoning
  about epistemic actions.  One of the most important works in this area is the
  work of Baltag, Moss and Solecki~\cite{baltag1998} which introduced the
  notion of action model logic, building on the earlier work of Gerbrandy and
  Groeneveld~\cite{gerbrandy1997}. In later work Baltag and Moss extended
  action model logic to consider epistemic programs~\cite{baltag2005} which are
  expressions built from action models using such operators as sequential
  composition, non-deterministic choice and iteration. The atoms of these
  programs are action models, so the approach is still inherently semantic in
  nature. The logic is unable to decompose the program beyond the level of the
  atoms, which themselves may be complex semantic objects.
  
  The relational actions of van Ditmarsch~\cite{vanditmarsch2001} provides a
  syntactic mechanism for describing an epistemic action, and provides the
  foundation for a lot of the work presented in this paper. The relational
  actions are constructed using essentially the same operators as in the
  language of action formulae. While the language is very similar, the
  semantics given are quite different~\cite{vanditmarsch2007}. In the logic of
  epistemic actions the semantics are given in such a way that worlds in a
  model are specified with respect to subsets of agents, so that the model is
  restricted to agents for whom the epistemic action was applied. The semantics
  were also specific to \classS{}, and non-trivial to generalise to other
  epistemic logics. A version of relational actions with concurrency is able to
  describe any \classS{} action model, although it is unknown whether the
  expressivity of concurrent relational actions is greater than that of action
  models~\cite{baltag2006}. Here we have generalised the approach and
  provided a correspondence theorem for action model logic. This has allowed us
  to retain the more familiar semantics of epistemic logic, generalise the
  logic to \classK{} and \classKFF{} as well as access existing
  synthesis results for dynamic epistemic logic~\cite{hales2013}.
  
  The synthesis result presented here is built on the work of
  Hales~\cite{hales2013} which gave a method to build an action model to
  satisfy a given epistemic goal. This construction inspired the syntactic
  description of epistemic actions and approach that we have used in this
  paper.
  
  Related synthesis results have been given by Aucher, et
  al.~\cite{aucher2011,aucher2012,aucher2013} which presents an event model
  language and uses it to give a thorough exploration of the relationship
  between epistemic models, action models and epistemic goals. Aucher defines a
  logic for action models and provides calculi to describe epistemic
  progression (what is true after executing a given action in a given model)
  epistemic regression (what is the most general precondition for an epistemic
  action given an epistemic goal) and epistemic planning (what action is
  sufficient to achieve an epistemic goal given some precondition). In future
  work we hope to extend the correspondence between action formula logic and
  action models to include Aucher's event model language.

\bibliographystyle{aiml14}
\bibliography{aiml14}

\end{document}